\documentclass[english]{IEEEtran}
\usepackage[T1]{fontenc}
\usepackage[latin9]{inputenc}
\usepackage{color}
\usepackage{float}
\usepackage{amsmath}
\usepackage{amsthm}
\usepackage{amssymb}
\usepackage{stackrel}
\usepackage{graphicx}
\usepackage{algorithm} %format of the algorithm
\usepackage{algorithmic} %format of the algorithm
\usepackage{multirow} %multirow for format of table
\usepackage{xcolor}
\usepackage{cite}
\usepackage{epstopdf}
\usepackage{changepage}
\usepackage{lipsum}

\makeatletter
\floatstyle{ruled}
\newfloat{algorithm}{tbp}{loa}
\providecommand{\algorithmname}{Algorithm}
\floatname{algorithm}{\protect\algorithmname}

\theoremstyle{plain}
\newtheorem{thm}{\protect\theoremname}
\theoremstyle{plain}
\newtheorem{prop}{\protect\propositionname}
\theoremstyle{remark}
\newtheorem{rem}{\protect\remarkname}
\theoremstyle{definition}
\newtheorem{defn}{\protect\definitionname}
\theoremstyle{plain}
\newtheorem{lem}{\protect\lemmaname}
\theoremstyle{plain}
\newtheorem{cor}{\protect\corollaryname}

\usepackage{amsfonts}
\usepackage{array}
\usepackage{algorithm}
\usepackage{algorithmic}
\usepackage{subfigure}
\usepackage{stfloats}

\makeatother

\usepackage{babel}

\makeatother

\usepackage{babel}
\providecommand{\corollaryname}{Corollary}
\providecommand{\definitionname}{Definition}
\providecommand{\lemmaname}{Lemma}
\providecommand{\propositionname}{Proposition}
\providecommand{\remarkname}{Remark}
\providecommand{\theoremname}{Theorem}

\usepackage{caption2}
\addto\captionsenglish{}

\begin{document}

\title{On Optimal Power Allocation for Downlink Non-Orthogonal Multiple
Access Systems}

\author{Jianyue~Zhu,~\IEEEmembership{Student~Member,~IEEE}, Jiaheng~Wang,~\IEEEmembership{Senior~Member,~IEEE},
Yongming~Huang,~\IEEEmembership{Member,~IEEE}, Shiwen~He,~\IEEEmembership{Member,~IEEE},
Xiaohu~You,~\IEEEmembership{Fellow,~IEEE}, and Luxi~Yang,~\IEEEmembership{Member,~IEEE}
\thanks{The authors are with the National Mobile Communications Research Lab,
Information Science and Engineering School, Southeast University,
Nanjing, China (email: \{zhujy, jhwang, huangym, shiwenhe, xhyu, lxyang\}@seu.edu.cn). }}
\maketitle
\begin{abstract}
Non-orthogonal multiple access (NOMA) enables power-domain multiplexing
via successive interference cancellation (SIC) and has been viewed
as a promising technology for 5G communication. The full benefit of
NOMA depends on resource allocation, including power allocation and
channel assignment, for all users, which, however, leads to mixed
integer programs. In the literature, the optimal power allocation
has only been found in some special cases, while the joint optimization
of power allocation and channel assignment generally requires exhaustive
search. In this paper, we investigate resource allocation in downlink
NOMA systems. As the main contribution, we analytically characterize
the optimal power allocation with given channel assignment over multiple
channels under different performance criteria. Specifically, we consider
the maximin fairness, weighted sum rate maximization, sum rate maximization
with quality of service (QoS) constraints, energy efficiency maximization
with weights or QoS constraints in NOMA systems. We also take explicitly
into account the order constraints on the powers of the users on each
channel, which are often ignored in the existing works, and show that
they have a significant impact on SIC in NOMA systems. Then, we provide
the optimal power allocation for the considered criteria in closed
or semi-closed form. We also propose a low-complexity efficient method
to jointly optimize channel assignment and power allocation in NOMA
systems by incorporating the matching algorithm with the optimal power
allocation. Simulation results show that the joint resource optimization
using our optimal power allocation yields better performance than
the existing schemes.
\end{abstract}


\section{Introduction}

As a result of the popularity of internet-of-things and cloud-based
applications, there is an explosive demand of new services and data
traffic for wireless communications. Hence, the fifth generation (5G)
communication systems propose higher requirements in data rates, lower
latency, and massive connectivity \cite{andrews2014will}. In order
to meet these high demands, some potential technologies, such as massive
multiple-input multiple-output (MIMO) \cite{jungnickel2014role},
millimeter wave \cite{7506342}, small cell
\cite{andrews2012femtocells,7835181,7155564} and device to device communication \cite{6510562,07572117}
will be introduced into 5G communication
systems. Recently, non-orthogonal multiple access (NOMA), which can
support overloaded transmission with limited resources and further
improve the spectral efficiency \cite{wei2016survey}, arises as a
promising technology for 5G communication systems.

The conventional multiple access schemes, which are categorized as
orthogonal multiple access technologies, are not sufficient to support
a massive connectivity because different users are allocated to orthogonal
resources in order to mitigate multiple access interference \cite{wang2006comparison}.
On the other side, by using superposition coding at the transmitter
with successive interference cancellation (SIC) at the receiver, NOMA
allows allocating one (frequency, time, code, or spatial) channel
to multiple users at the same time \cite{saito2013non}, which can
lead to better performance in terms of spectral efficiency, fairness,
or energy efficiency \cite{dai2015non}. Therefore, NOMA has received
much attention recently. In \cite{ding2015general} and \cite{ding2016application},
the authors discussed an combination of NOMA with MIMO technologies.
NOMA has also been introduced to be used with other technologies,
e.g., visible light communication \cite{2016userandpower} and millimeter
wave communication \cite{7862785}.

The basic idea of NOMA is to implement multiple access in the power
domain \cite{wei2016survey}. Hence, the key to achieve the full benefit
of NOMA systems is resource allocation, which usually include power
allocation and channel assignment. Unfortunately, the joint optimization
of power allocation and channel assignment in NOMA systems leads to
a mixed integer program \cite{zhang2016radio}, which has been proved
to be a NP-hard problem in \cite{lei2015joint}. Hence, finding the
jointly optimal resource allocation generally requires exhaustive
search \cite{sun2016optimal}, which, however, causes prohibitive
complexity and is not applicable for practical systems. Therefore,
suboptimal but efficient resource optimization methods are more preferred
in practice. Such efficient methods are often obtained by optimizing
power allocation and channel assignment alternately \cite{parida2014power,fang2016energy,zhang2016radio,hojeij2015resource}.
In this paper, we investigate resource allocation with a focus on
power allocation for downlink NOMA systems under various criteria.

\subsection{Related Works}

In NOMA systems, resource allocation has been studied for different
performance measures. In the literature, the sum rate maximization
is the most commonly adopted objective, and there are a number of
related works \cite{wang2016power,parida2014power,sun2016optimal,hojeij2015resource}.
In \cite{wang2016power}, the authors investigated the optimal power
allocation to maximize the sum rate with QoS constraints only for
two users on one channel. In \cite{parida2014power}, the problem
of maximizing the weighted sum rate in a downlink orthogonal frequency
division multiplexing access (OFDMA) based NOMA system was studied,
where the nonconvex power allocation problem was solved via DC (difference
of two convex functions) programming and thus only a suboptimal power
allocation solution was provided. In \cite{sun2016optimal}, the authors
also considered the weighted sum rate maximization and exploited monotonic
optimization to develop an optimal joint power allocation and channel
assignment policy, which, however, has an exponential complexity and
only serves as a system performance benchmark. In \cite{hojeij2015resource},
the authors introduced a resource allocation method based on waterfilling
to improve the total achieved system throughput but there is no guarantee
for the optimality of the obtained solution.

Fairness is also an important issue in NOMA systems, where the most
common fairness indication is the maximin fairness (MMF). Therefore,
a number of works has studied resource allocation for MMF, e.g., \cite{cui2016novel,choi2016power,timotheou2015fairness}.
In \cite{cui2016novel} and \cite{timotheou2015fairness}, the authors
investigated the optimal power allocation based on MMF for users on
one channel using statistical channel state information (CSI) and
instantaneous CSI, respectively. The proportional fairness scheduling
that maximizes the weighted MMF was studied in \cite{choi2016power},
where the optimal solution was only derived for two users on a single
channel.

As energy efficiency (EE) becomes an important performance measure
of wireless communication systems, the resource problem that maximizes
the EE in NOMA systems has also been considered but only in two works
\cite{zhang2016energy,fang2016energy}. In \cite{zhang2016energy},
the authors developed the optimal power allocation for maximizing
the EE with QoS constraints but only for the users on one channel.
The joint power allocation and channel assignment for maximizing the
EE was considered in \cite{fang2016energy}, whereas only a suboptimal
solution was obtained via DC programming.

In summary, so far the optimal power allocation was only found for
users on a single channel under particular performance criteria, but
unknown in the general case for all users on multiple channels. Furthermore,
in almost all existing works, the order constraints on the powers
of users were either ignored or not explicitly taken into account.

\subsection{Contributions}

In this paper, we investigate resource allocation in downlink NOMA
systems with a focus on seeking the optimal power allocation for multiple
channels and users under various performance criteria. The contributions
in this paper are summarized in the following:
\begin{itemize}
\item We consider different criteria that lead to different problem formulations,
including the maximin fairness, the weighted sum rate maximization,
the sum rate maximization with QoS constraints, the energy efficiency
maximization with weights or QoS constraints.
\item We take explicitly into account the order constraints on the powers
of users on each channel that guarantee the decoding order of SIC
on each channel unchanged in NOMA systems.
\item Then, we analytically characterize the optimal power allocation and
provide closed-form or semi-closed solutions to the formulated power
optimization problems.
\item It is shown that the power order constraints could result in an equal
signal strength, which may cause a failure of SIC or a large error
propagation. Thus, we introduce the concept of SIC-stability and identify
the conditions that avoid equal power allocation in NOMA systems under
different criteria.
\item We propose an efficient method to jointly optimize the channel assignment
and power allocation by incorporating the matching algorithm with
our optimal power allocation and iteratively using them to refine
the solution.
\item The obtained optimal power allocation can also be used with other
channel assignment algorithms and can even reduce the complexity of
the exhaustive search for jointly optimal resource allocation.
\item Finally, it is shown via simulations that the proposed joint resource
optimization method outperforms the existing schemes and achieves
near-optimal performance.
\end{itemize}
\textcolor{blue}{\vspace{-0.5cm}
}

The rest of the paper is organized as follows. Section \ref{sec2}
introduces the NOMA system model and various resource optimization
problems under different performance criteria and constraints. In
Section \ref{secMMF}, Section \ref{sec:WSR}, and Section \ref{sec:GEE},
we investigate the optimal power allocation for the MMF, sum rate
maximization, and EE maximization, respectively. In Section \ref{sec:MATCHING},
a joint channel assignment and power allocation optimization algorithm
is proposed. The performance of the proposed power allocation is evaluated
in section \ref{sec:SIMULATION} by simulations and the conclusion
is drawn in Section \ref{sec:CONCLUTION}.

\section{\label{sec2}Problem Statement}

\subsection{System Model }

Consider a downlink NOMA network wherein a \textcolor{black}{base
station (BS)}\textcolor{blue}{{} }serves $N$ users through $M$ channels.
The total bandwidth $B$ is equally divided to $M$ channels so the
bandwidth of each channel is $B_{c}=B/M$. Let $N_{m}\in\{N_{1},N_{2},...,N_{M}\}$
be the number of users using channel $m$ for $m=1,2,\cdots,M$ and
$\textrm{U}\textrm{E}_{n,m}$ denotes user $n$ on channel $m$ for
$n=1,2,\cdots,N_{m}$. The signal transmitted by the BS on each channel
$m$ can be expressed as
\[
x_{m}=\stackrel[n=1]{N_{m}}{\sum}\sqrt{p_{n,m}}s_{n}
\]
where $s_{n}$ is the symbol of $\textrm{U}\textrm{E}_{n,m}$ and
$p_{n,m}$ is the power allocated to $\textrm{U}\textrm{E}_{n,m}$.
The received signal at $\textrm{U}\textrm{E}_{n,m}$ is
\[
y_{n,m}=\sqrt{p_{n,m}}h_{n,m}s_{n}+\stackrel[i=1,i\neq n]{N_{m}}{\sum}\sqrt{p_{i,m}}h_{n,m}s_{i}+z_{n,m}
\]
where $h_{n,m}=g_{n,m}d_{n}^{-\alpha}$ is the channel coefficient
from the BS to $\textrm{U}\textrm{E}_{n,m}$, $g_{n,m}$ follows a
Rayleigh distribution, $d_{n}$ is the distance between the BS and
$\textrm{U}\textrm{E}_{n,m}$, $\alpha$ is the path-loss exponent,
and $z_{n,m}\sim\mathcal{CN}(0,\sigma_{m}^{2})$ is the additive white
Gaussian noise (AWGN).

According to the principle of NOMA, one channel can be assigned to
multiple users, who will use SIC to decode their signals. \textcolor{black}{Specifically,
let $\Gamma_{n,m}=\left|h_{n,m}\right|^{2}/\sigma_{m}^{2}$ be the
channel to noise ratio (CNR) of $\textrm{U}\textrm{E}_{n,m}$. }Assume
without loss of generality (w.l.o.g.) that the CNRs of the users on
channel $m$ are ordered as
\[
\Gamma_{1,m}\geq\cdots\geq\Gamma_{n,m}\geq\cdots\geq\Gamma_{N_{m},m}
\]
i.e., $\textrm{U}\textrm{E}_{1,m}$ and $\textrm{U}\textrm{E}_{N_{m},m}$
are the strongest and weakest users on channel $m$, respectively.
Then, the NOMA protocol allocates higher powers to the users with
lower CNRs \cite{wei2016survey,ding2015cooperative}, leading to $p_{1,m}\leq\cdots\leq p_{n,m}\leq\cdots\leq p_{N_{m},m}$.
Hence, $\textrm{U}\textrm{E}_{n,m}$ is able to decode signals of
$\textrm{U}\textrm{E}_{i,m}$ for $i>n$ and remove them from its
own signal, but treats the signals from $\textrm{U}\textrm{E}_{i,m}$
for $i<n$ as interference. Therefore, the signal to interference-plus-noise
ratio (SINR) of $\textrm{U}\textrm{E}_{n,m}$ using SIC is given by
\[
\gamma_{n,m}=\frac{p_{n,m}\Gamma_{n,m}}{1+\sum_{i=1}^{n-1}p_{i,m}\Gamma_{n,m}}.
\]
Thus, the data rate of $\textrm{U}\textrm{E}_{n,m}$ is
\[
R_{n,m}(p_{n,m})=B_{c}\log\left(1+\frac{p_{n,m}\Gamma_{n,m}}{1+\sum_{i=1}^{n-1}p_{i,m}\Gamma_{n,m}}\right).
\]

Using SIC at each user's receiver causes additional complexity, which
is proportional to the number of users on the same channel. Thus,
in practice, each channel is often restricted to be assigned to two
users \cite{ding2016impact,ding2015cooperative,fang2016energy}, which
is also beneficial to reduce the error propagation of SIC. In this
paper, we would also like to focus on this typical situation and assume
that $N_{m}=2$ for $m=1,2,\cdots,M$ and $N=2M$. In this case, suppose
w.l.o.g. that the CNRs of $\textrm{U}\textrm{E}_{1,m}$ and $\textrm{U}\textrm{E}_{2,m}$
are ordered as $\Gamma_{1,m}\geq\Gamma_{2,m}$. Then, the rates of
$\textrm{U}\textrm{E}_{1,m}$ and $\textrm{U}\textrm{E}_{2,m}$ on
channel $m$ are given respectively by
\[
\begin{array}{l}
R_{1,m}=B_{c}\log\left(1+p_{1,m}\Gamma_{1,m}\right),\\
R_{2,m}=B_{c}\log\left(1+\frac{p_{2,m}\Gamma_{2,m}}{p_{1,m}\Gamma_{2,m}+1}\right).
\end{array}
\]

\subsection{Problem Formulation\label{subsec:Problem-Formulation}}

The performance of a NOMA scheme relies on resource allocation, including
power allocation and channel assignment, for all users. In this paper,
we investigate optimization of resource allocation for NOMA systems.
For this purpose, we consider the following performance measures.
\begin{enumerate}
\item \textit{Maximin fairness}: A common criterion is the maximin fairness
(MMF), which aims to provide fairness for all users. The corresponding
resource allocation problem is given by
\begin{equation}
\max\underset{m=1,\ldots,M}{\min}~\left\{ R_{1,m},R_{2,m}\right\} .\label{maximin fairness}
\end{equation}
The similar problems have been studied in \cite{choi2016power,cui2016novel,timotheou2015fairness},
whereas the optimal power allocation was only found for a few users
on a single channel but unknown for all users over multiple channels.
\item \textit{Sum rate: }The most common objective is to maximize the sum
rate (SR) of all users. To avoid that the resource on each channel
is occupied by one user, weights or QoS constraints are often introduced
into SR maximization. In this paper, we consider both the weighted
SR maximization:
\begin{equation}
\max\ \sum_{m=1}^{M}\left(W_{1,m}R_{1,m}+W_{2,m}R_{2,m}\right)\label{max WSR}
\end{equation}
where $W_{n,m}$ is the weight of $\textrm{U}\textrm{E}_{n,m}$, and
the SR maximization with QoS constraints:
\begin{align}
\max~ & \sum_{m=1}^{M}\left(R_{1,m}+R_{2,m}\right)\label{max QSR}\\
\textrm{s.t.\ } & R_{n,m}\geq R_{n,m}^{\min},\ n=1,2,\ \forall m\nonumber
\end{align}
where $R_{n,m}^{\min}$ is the QoS threshold of $\textrm{U}\textrm{E}_{n,m}$.
Although two problems have been studied in a number of works, e.g.,
\cite{sun2016optimal,wang2016power,parida2014power,hojeij2015resource},
the optimal power allocation was only found for two users on one channel
\cite{wang2016power}, while the joint resource optimization is either
suboptimal \cite{parida2014power,hojeij2015resource} or needs exhaustive
search \cite{sun2016optimal}.
\item \textit{Energy efficiency}: In this paper, we also consider improving
energy efficiency (EE) of the NOMA system, which is defined as the
ratio between the sum rate and the power consumption of the whole
system. When weights or QoS constraints are introduced, the EE maximization
problem can be formulated as
\begin{align}
\max~ & \frac{\sum_{m=1}^{M}\left(W_{1,m}R_{1,m}+W_{2,m}R_{2,m}\right)}{P_{T}+\sum_{m=1}^{M}\left(p_{1,m}+p_{2,m}\right)}\label{max WEE}
\end{align}
or
\begin{align}
\max~ & \frac{\sum_{m=1}^{M}\left(R_{1,m}+R_{2,m}\right)}{P_{T}+\sum_{m=1}^{M}\left(p_{1,m}+p_{2,m}\right)}\label{max QEE}\\
\textrm{s.t.\ } & R_{n,m}\geq R_{n,m}^{\min},\ n=1,2,\ \forall m\nonumber
\end{align}
where $P_{T}$ is the power consumption of the circuits and SIC on
all channels. This problem has only been studied in \cite{zhang2016energy}
but only for one channel.
\end{enumerate}
\textcolor{blue}{\vspace{-0.5cm}
}

In addition to the above objectives and QoS constraints, one shall
also consider power constraints in NOMA systems. The transmit power
constraint of the BS is given by
\[
\sum_{m=1}^{M}(p_{1,m}+p_{2,m})\leq P
\]
where $P$ is the total power budget of the BS. In NOMA systems, there
is an implicit power constraint for the users on each channel $m$,
i.e .,
\[
p_{1,m}\leq p_{2,m},\ m=1,\ldots,M
\]
which is to guarantee that a higher power is allocated to the user
with a lower CNR (i.e.,$\textrm{U}\textrm{E}_{2,m}$) on channel $m$
so that the decoding order of the SIC is not changed. However, in
most existing works, the power order constraints were ignored. In
this paper, we will show that it is important to take such constraints
into account explicitly in power allocation for NOMA.

The joint optimization of power allocation and channel assignment
in NOMA systems is, unfortunately, a mixed integer problem. Finding
the jointly optimal solution requires exhaustive search \cite{zhang2016radio},
which results in prohibitive computational complexity. Therefore,
in practice, power allocation and channel assignment are often separately
and alternatively optimized, i.e., fix one and optimize the other
\cite{fang2016energy,parida2014power,zhang2016radio}, which may lead
to, though possibly suboptimal, efficient resource allocation solutions.
In this paper, we would also like to use this methodology. Specifically,
we first optimize power allocation with given channel assignment,
and then optimize channel assignment. The most exciting thing is that,
different from all existing works, we are able to find the optimal
power allocation for all users over multiple channels for all above
considered performance measures. The optimal power allocation is either
given in a closed-form expression or can be efficiently obtained via
the proposed algorithms. Our results will dramatically simplify the
joint resource allocation and improve the system performance.

\section{\label{secMMF}Optimal Power Allocation for Maximin Fairness}

The NOMA scheme enables a flexible management of the users' achievable
rates and provides an efficient way to enhance user fairness. In this
section, we study the optimal power allocation to achieve the maximin
fairness (MMF) in the NOMA system. According to (\ref{maximin fairness}),
with given channel assignment, the MMF problem is equivalent to the
following power allocation problem:
\[
\mathcal{OP}_{1}^{\textrm{MMF}}\!\!:\!\!\begin{array}{c}
\underset{\boldsymbol{p_{1},p_{2}}}{\max}\\
\textrm{s.t.}
\end{array}\!\begin{array}{l}
\!\!\!\!\!\underset{m=1,\ldots,M}{\min}\!\left\{ R_{1,m}(p_{1,m},p_{2,m}),R_{2,m}(p_{1,m},p_{2,m})\right\} \\
\boldsymbol{0}\leq\boldsymbol{p}_{1}\leq\boldsymbol{p}_{2},\sum_{m=1}^{M}p_{1,m}+p_{2,m}\leq P
\end{array}
\]
where $\boldsymbol{p}_{1}=\{p_{1,m}\}_{m=1}^{M}$ and $\boldsymbol{p}_{2}=\{p_{2,m}\}_{m=1}^{M}$.
However, $\mathcal{OP}_{1}^{\textrm{MMF}}$ is a nonconvex problem,
as its objective is not concave. Its optimal solution has only been
found in the special case $M=1$\cite{choi2016power,cui2016novel,timotheou2015fairness},
i.e., a single channel, but unknown in the general case yet.

To address this problem, we first introduce auxiliary variables $\boldsymbol{q}=\{q_{m}\}_{m=1}^{M}$,
where $q_{m}$ represents the power budget for channel $m$ with $p_{1,m}+p_{2,m}=q_{m}$.
Suppose that the channel power budgets $\{q_{m}\}_{m=1}^{M}$ are
given. Then, $\mathcal{OP}_{1}^{\textrm{MMF}}$ is decomposed into
a group of subproblems for each channel $m$:
\[
\mathcal{OP}_{2,m}^{\textrm{MMF}}\!:\!\!\begin{array}{cl}
\underset{p_{1,m},p_{2,m}}{\max} & \!\!\!\!\!\min\left\{ R_{1,m}(p_{1,m},p_{2,m}),R_{2,m}(p_{1,m},p_{2,m})\right\} \\
\textrm{s.t.} & 0\leq p_{1,m}\leq p_{2,m},p_{1,m}+p_{2,m}=q_{m}.
\end{array}
\]
We first solve subproblem $\mathcal{OP}_{2,m}^{\textrm{MMF}}$ and
show that its optimal solution is given in a closed form.
\begin{prop}
\label{P1 MMF}Suppose that $\Gamma_{1,m}\geq\Gamma_{2,m}$. Then,
the optimal solution to $\mathcal{OP}_{2,m}^{\textrm{MMF}}$ is given
by $p_{1,m}^{\textrm{ }\star}=\Lambda_{m}$ and $p_{2,m}^{\star}=q_{m}-p_{1,m}^{\star}$,
where $\Gamma_{l,m}\triangleq\left|h_{l,m}\right|^{2}/\sigma_{m}^{2}$
and
\[
\Lambda_{m}\triangleq\frac{-\left(\Gamma_{1,m}+\Gamma_{2,m}\right)\!+\!\sqrt{\left(\Gamma_{1,m}+\Gamma_{2,m}\right){}^{2}+4\Gamma_{1,m}\Gamma_{2,m}^{2}q_{m}}}{2\Gamma_{1,m}\Gamma_{2,m}}.
\]
\end{prop}
\begin{proof} From the constraint $p_{1,m}+p_{2,m}=q_{m}$, we have
$p_{2,m}=q_{m}-p_{1,m}$, so $p_{1,m}\leq p_{2,m}$ is equivalent
to $p_{1,m}\leq q_{m}/2$. Substituting $p_{2,m}=q_{m}-p_{1,m}$ into
$R_{1,m}$ and $R_{2,m}$, we obtain
\[
\begin{array}{l}
R_{1,m}(p_{1,m})\triangleq B_{c}\log\left(1+p_{1,m}\Gamma_{1,m}\right),\\
R_{2,m}(p_{1,m})\triangleq B_{c}\log\left(\frac{q_{m}\Gamma_{2,m}+1}{p_{1,m}\Gamma_{2,m}+1}\right).
\end{array}
\]
If $p_{1,m}\geq\Lambda_{m}$, then $R_{1,m}(p_{1,m})\geq R_{2,m}(p_{1,m})$
and the objective of $\mathcal{OP}_{2,m}^{\textrm{MMF}}$ is $R_{2,m}(p_{1,m})$,
which is decreasing in $p_{1,m}$. So the maximizer is the lower bound
$p_{1,m}=\Lambda_{m}$. If $p_{1,m}\leq\Lambda_{m}$, then $R_{1,m}(p_{1,m})\leq R_{2,m}(p_{1,m})$
and the objective of $\mathcal{OP}_{2,m}^{\textrm{MMF}}$ is $R_{1,m}(p_{1,m})$,
which is increasing in $p_{1,m}$. So the maximizer is the upper bound
$p_{1,m}=\Lambda_{m}$. Therefore, the optimal point is $p_{1,m}^{\star}=\Lambda_{m}$.
Finally, it can be verified that $p_{1,m}^{\star}=\Lambda_{m}\leq q_{m}/2$
. \end{proof}
\begin{rem}
From Proposition \ref{P1 MMF}, we obtain $R_{1,m}(p_{1,m}^{\star},p_{2,m}^{\star})=R_{2,m}(p_{1,m}^{\star},p_{2,m}^{\star})=f_{m}^{\textrm{MMF}\star}$,
where
\begin{multline}
\!\!\!\!\!f_{m}^{\textrm{MMF}\star}\triangleq\\
B_{c}\log\!\!\left(\!\!\frac{\Gamma_{2,m}\!-\!\Gamma_{1,m}\!+\!\sqrt{\left(\Gamma_{1,m}\!+\!\Gamma_{2,m}\right){}^{2}\!+\!4\Gamma_{1,m}\Gamma_{2,m}^{2}q_{m}}}{2\Gamma_{2,m}}\!\right)\!\!,\label{R_1m*,R_2m*-MMF}
\end{multline}
i.e., $\textrm{U}\textrm{E}_{1,m}$ and $\textrm{U}\textrm{E}_{2,m}$
achieve the same rate at the optimal point. This indicates that, under
the MMF criterion, the NOMA system will provide absolute fairness
for two users on one channel.

To elaborate another important insight, we introduce the following
definition.
\end{rem}
\begin{defn}
\label{SIC-stable}A NOMA system is called \emph{SIC-stable} if the
optimal power allocation satisfies $p_{1,m}<p_{2,m}$ on each channel
$m$.
\end{defn}
\begin{rem}
\textcolor{black}{In NOMA systems, SIC is performed according to the
order of the CNRs of the users on one channel \cite{ding2015cooperative,wei2016survey},
which is guaranteed by imposing an inverse order of the powers allocated
to the users, i.e., $p_{1,m}\leq p_{2,m}$ on channel $m$. Specifically,
$\textrm{U}\textrm{E}_{1,m}$ (the strong user with a higher CNR)
first decodes the signal of $\textrm{U}\textrm{E}_{2,m}$ (the weak
user with a lower CNR) and then subtracts it from the superposed signal.
Therefore, from the SIC perspective, a difference between the signal
strengths of $\textrm{U}\textrm{E}_{2,m}$ and $\textrm{U}\textrm{E}_{1,m}$
is necessary \cite{ali2016dynamic}. However, even with the power
order constraint, the power optimization may lead to $p_{1,m}=p_{2,m}$,
i.e., $\textrm{U}\textrm{E}_{1,m}$ and $\textrm{U}\textrm{E}_{2,m}$
have the same signal strength, which is the worst situation for SIC.
In this case, SIC may fail or has a large error propagation and thus
is unstable. Indeed, the authors in \cite{6868214} pointed out that
the power of the weak user must be strictly larger than that of the
strong user, otherwise the users' outage probabilities will always
be one. Definition \ref{SIC-stable} explicitly concretizes such a
practical requirement in NOMA systems.}
\end{rem}
\begin{lem}
\label{lem:mmf_stable}The NOMA system is SIC-stable for $\mathcal{OP}_{1}^{\textrm{MMF}}$.
\end{lem}
\begin{proof} Given $\Gamma_{1,m}\geq\Gamma_{2,m}$, we have
\begin{align*}
\Lambda_{m} & =\frac{2\Gamma_{2,m}q_{m}}{\left(\Gamma_{1,m}+\Gamma_{2,m}\right)+\sqrt{\left(\Gamma_{1,m}+\Gamma_{2,m}\right){}^{2}+4\Gamma_{1,m}\Gamma_{2,m}^{2}q_{m}}}\\
 & <\frac{\Gamma_{2,m}q_{m}}{\Gamma_{1,m}+\Gamma_{2,m}}\leq\frac{q_{m}}{2},
\end{align*}
which indicates $p_{1,m}^{\star}<p_{2,m}^{\star}$ for each $m$.
Therefore, the NOMA system is SIC-stable. \end{proof}
\begin{rem}
\textcolor{black}{According to Definition \ref{SIC-stable} and indicated
by Lemma \ref{lem:mmf_stable}, the NOMA system is always SIC-stable
under the MMF criterion, as in this case the optimal power allocation
always satisfies $p_{1,m}^{\star}<p_{2,m}^{\star}$, $\forall m$.
On the other hand, in the subsequent sections, we will show that a
NOMA system is not always SIC-stable under different criteria and
constraints.}
\end{rem}
To obtain the optimal power allocation for all channels, we shall
optimize the power budget $q_{m}$ for each channel $m$. According
to $\mathcal{OP}_{1}^{\textrm{MMF}}$ and $\mathcal{OP}_{2,m}^{\textrm{MMF}}$,
the corresponding power budget optimization problem is given by
\[
\mathcal{OP}_{3}^{\textrm{MMF}}:\begin{array}{cl}
\underset{\boldsymbol{q}}{\max}~ & \underset{m=1,\ldots,M}{\min}\ f_{m}^{\textrm{MMF}\star}(q_{m})\\
\textrm{s.t.} & \sum_{m=1}^{M}q_{m}\leq P,\;\boldsymbol{q}\geq\boldsymbol{0}
\end{array}
\]
where $f_{m}^{\textrm{MMF}\star}(q_{m})$ is the optimal objective
value of $\mathcal{OP}_{2,m}^{\textrm{MMF}}$ and given in \eqref{R_1m*,R_2m*-MMF}.
\begin{lem}
\label{lem:mmf_concave}$f_{m}^{\textrm{MMF}\star}(q_{m})$ is a concave
function.
\end{lem}
\begin{proof} It can be verified that $\partial^{2}f_{m}^{\textrm{MMF}\star}/\partial q_{m}^{2}<0$
and hence $f_{m}^{\textrm{MMF}\star}(q_{m})$ is concave. \end{proof}
From Lemma \ref{lem:mmf_concave}, $\mathcal{OP}_{3}^{\textrm{MMF}}$
is actually a convex problem, whose solution can be efficiently found
via standard convex optimization tools, e.g., CVX. Nevertheless, we
are able to analytically characterize the optimal solution to \eqref{max t}.
\begin{thm}
\label{T1MMF}The optimal solution to $\mathcal{OP}_{3}^{\textrm{MMF}}$
is given by
\begin{equation}
q_{m}^{\star}=\frac{\left(Z\left(\lambda\right)\Gamma_{2,m}+\Gamma_{1,m}\right)\left(Z\left(\lambda\right)-1\right)}{\Gamma_{1,m}\Gamma_{2,m}},\ \forall m\label{optimal q_m MMF}
\end{equation}
where
\begin{align*}
Z\left(\lambda\right)\triangleq & X+\sqrt{X^{2}+\frac{B_{c}}{2\lambda\sum_{m=1}^{M}1/\Gamma_{1,m}}},\\
X\triangleq & \frac{\sum_{m=1}^{M}\left(\Gamma_{2,m}-\Gamma_{1,m}\right)/\left(\Gamma_{1,m}\Gamma_{2,m}\right)}{4\sum_{m=1}^{M}1/\Gamma_{1,m}}
\end{align*}
and $\lambda$ is chosen such that $\sum_{m=1}^{M}q_{m}^{\star}=P$.
\end{thm}
\begin{proof} We first transform $\mathcal{OP}_{3}^{\textrm{MMF}}$
into
\begin{equation}
\begin{array}{cl}
\underset{\boldsymbol{q},t}{\max} & t\\
\textrm{s.t.} & \boldsymbol{q}\geq\boldsymbol{0,}\ \sum_{m=1}^{M}q_{m}\leq P,\ f_{m}^{\textrm{MMF}\star}(q_{m})\geq t,\forall m
\end{array}\label{max t}
\end{equation}
where $f_{m}^{\textrm{MMF}\star}(q_{m})\geq t$ is equivalent to $q_{m}\geq(a^{t}\Gamma_{2,m}+\Gamma_{1,m})(a^{t}-1)/(\Gamma_{1,m}\Gamma_{2,m})$
with $a=2^{1/B_{c}}$. Then, the Lagrange of (\ref{max t}) can be
written as
\begin{align*}
L= & t+\sum_{m=1}^{M}\mu_{m}\left[q_{m}-\frac{\left(a^{t}\Gamma_{2,m}+\Gamma_{1,m}\right)\left(a^{t}-1\right)}{\Gamma_{1,m}\Gamma_{2,m}}\right]\\
 & -\lambda\left(\sum_{m=1}^{M}q_{m}-P\right)
\end{align*}
where $\left\{ \mu_{m}\right\} _{m=1}^{M}$ and $\lambda$ are the
Lagrange multipliers. Since (\ref{max t}) is a convex optimization
problem, its optimal solution is characterized by the following Karush-Kuhn-Tucker
(KKT) conditions:
\begin{equation}
\frac{\partial L}{\partial q_{m}}=\mu_{m}-\lambda=0,\label{1}
\end{equation}
\begin{equation}
\!\!\frac{\partial L}{\partial t}\!=\!1\!-\!a^{2t}\!\sum_{m=1}^{M}\frac{2\mu_{m}\ln a}{\Gamma_{1,m}}+\!a^{t}\!\sum_{m=1}^{M}\frac{\mu_{m}\ln a}{\Gamma_{1,m}}-\!a^{t}\!\sum_{m=1}^{M}\frac{\mu_{m}\ln a}{\Gamma_{2,m}}\!=\!0,\label{2}
\end{equation}
\begin{equation}
\mu_{m}\left(q_{m}-\frac{\left(a^{t}\Gamma_{2,m}+\Gamma_{1,m}\right)\left(a^{t}-1\right)}{\Gamma_{1,m}\Gamma_{2,m}}\right)=0,\label{3}
\end{equation}
\[
\lambda\left(\sum_{m=1}^{M}q_{m}-P\right)=0.
\]
It follows from \eqref{1} and \eqref{2} that $\mu_{m}=\lambda\neq0$.
Then, \eqref{2} is equivalent to
\[
C_{1}a^{2t}-C_{2}a^{t}-1=0
\]
where $C_{1}=2\lambda\ln a\sum_{m=1}^{M}1/\Gamma_{1,m}$ and $C_{2}=\lambda\ln a\sum_{m=1}^{M}(\Gamma_{2,m}-\Gamma_{1,m})/(\Gamma_{1,m}\Gamma_{2,m})$.
By solving this quadratic equation, we obtain
\begin{align*}
a^{t} & =\frac{C_{2}}{2C_{1}}+\sqrt{\left(\frac{C_{2}}{2C_{1}}\right)^{2}+\frac{1}{C_{1}}}\\
 & =X+\sqrt{X^{2}+\frac{B_{c}}{2\lambda\sum_{m=1}^{M}1/\Gamma_{1,m}}}=Z\left(\lambda\right).
\end{align*}
Finally, from \eqref{3}, we have
\begin{align}
q_{m} & =\frac{\left(a^{t}\Gamma_{2,m}+\Gamma_{1,m}\right)\left(a^{t}-1\right)}{\Gamma_{1,m}\Gamma_{2,m}}\nonumber \\
 & =\frac{\left(Z\left(\lambda\right)\Gamma_{2,m}+\Gamma_{1,m}\right)\left(Z\left(\lambda\right)-1\right)}{\Gamma_{1,m}\Gamma_{2,m}}\geq0\label{optimal q_m}
\end{align}
which completes the proof. \end{proof}
\begin{cor}
\textcolor{black}{Under the MMF criterion, the optimal power allocation
achieves the absolute fairness for all the users on all channels,
i.e., $R_{1,m}=R_{2,m}=r$, $m=1,\ldots,M$, for some $r\geq0$.}
\end{cor}
\begin{proof} \textcolor{black}{Given the optimal $q_{m}$ in (\ref{optimal q_m}),
it can be verified that $f_{m}^{\textrm{MMF}\star}(q_{m})=t$ for
$m=1,\ldots,M$. } \end{proof} The optimal power allocation under
the MMF criterion is fully characterized by Theorem \ref{T1MMF} and
Proposition \ref{P1 MMF}. It follows from (\ref{optimal q_m MMF})
that $q_{m}^{\star}$ is monotonically decreasing in $\lambda$, so
the optimal $\lambda$ satisfying $\sum_{m=1}^{M}q_{m}^{\star}=P$
can be efficiently found via a simple bisection method.

\section{\label{sec:WSR}Optimal Power Allocation for Sum Rate}

In this section, we seek the optimal power allocation for maximizing
the weighted sum rate (SR) or maximizing the SR with QoS constraints.

\subsection{\label{subsec:WSR1}Weighted SR Maximization (SR1)}

According to \eqref{max WSR}, with given channel assignment, the
problem of maximizing the weighted sum rate is equivalent to the following
power allocation problem:
\[
\mathcal{OP}_{1}^{\textrm{SR}1}:\begin{array}{cl}
\underset{\boldsymbol{p}_{1},\boldsymbol{p}_{2}}{\max} & \sum_{m=1}^{M}g(p_{1,m},p_{2,m})\\
\textrm{s.t.} & \boldsymbol{0}\leq\boldsymbol{p}_{1}\leq\boldsymbol{p}_{2},\ \sum_{m=1}^{M}\left(p_{1,m}+p_{2,m}\right)\leq P
\end{array}
\]
where $g(p_{1,m},p_{2,m})\triangleq W_{1,m}R_{1,m}(p_{1,m},p_{2,m})+W_{2,m}R_{2,m}(p_{1,m},p_{2,m}).$
As the objective of $\mathcal{OP}_{1}^{\textrm{SR}1}$ is not a concave
function, $\mathcal{OP}_{1}^{\textrm{SR}1}$ is also a nonconvex problem.
Although this problem has been studied in \cite{parida2014power,hojeij2015resource,sun2016optimal},
the solution is either suboptimal or needs exhaustive search.

Introduce auxiliary variables $\boldsymbol{q}=\{q_{m}\}_{m=1}^{M}$
that represent the power budgets on each channel $m$ with $p_{1,m}+p_{2,m}=q_{m}$.
Then, $\mathcal{OP}_{1}^{\textrm{SR}1}$ is decomposed into a group
of subproblems for each channel $m$:
\[
\mathcal{OP}_{2,m}^{\textrm{SR}1}:\begin{array}{cl}
\underset{p_{1,m},p_{2,m}}{\max} & g(p_{1,m},p_{2,m})\\
\textrm{s.t.} & 0\leq p_{1,m}\leq p_{2,m},\ p_{1,m}+p_{2,m}=q_{m}.
\end{array}
\]
We first solve the subproblem $\mathcal{OP}_{2,m}^{\textrm{SR}1}$
for each channel $m$. Note that $\mathcal{OP}_{2,m}^{\textrm{SR}1}$
is still a nonconvex problem due to the interference between $\textrm{U}\textrm{E}_{1,m}$
and $\textrm{U}\textrm{E}_{2,m}$. Nevertheless, its optimal solution
can be characterized in a closed form.
\begin{prop}
\label{P1 WSR1}Suppose that $\Gamma_{1,m}\geq\Gamma_{2,m}$, $1<W_{2,m}\text{/}W_{1,m}<\Gamma_{1,m}/\Gamma_{2,m}$
and $q_{m}>2\Omega_{m}$, with
\[
\Omega_{m}\triangleq\frac{W_{2,m}\Gamma_{2,m}-W_{1,m}\Gamma_{1,m}}{\Gamma_{1,m}\Gamma_{2,m}\left(W_{1,m}-W_{2,m}\right)}.
\]
Then, the optimal solution to $\mathcal{OP}_{2,m}^{\textrm{SR}1}$
is given by $p_{1,m}^{\star}=\Omega_{m}$ and $p_{2,m}^{\star}=q_{m}-p_{1,m}^{\star}$.
\end{prop}
\begin{proof} Since $p_{2,m}=q_{m}-p_{1,m}$, $p_{1,m}\leq p_{2,m}$
is equal to $p_{1,m}\leq q_{m}/2$, and the objective becomes
\begin{align*}
F(p_{1,m})\triangleq & W_{1,m}B_{c}\log\left(1+p_{1,m}\Gamma_{1,m}\right)\\
 & +W_{2,m}B_{c}\log\left(\frac{q_{m}\Gamma_{2,m}+1}{p_{1,m}\Gamma_{2,m}+1}\right).
\end{align*}
By setting the derivative of $F$ to zero, we have
\[
\frac{dF}{dp_{1,m}}=\frac{W_{1,m}B_{c}}{1/\Gamma_{1,m}+p_{1,m}}-\frac{W_{2,m}B_{c}}{1/\Gamma_{2,m}+p_{1,m}}=0,
\]
leading to a unique root $p_{1,m}=\Omega_{m}$, which satisfies the
constraint $p_{1,m}\leq q_{m}/2$ since $\Omega_{m}<q_{m}/2$. Given
$\Gamma_{1,m}\geq\Gamma_{2,m}$ and $1<W_{2,m}\text{/}W_{1,m}<\Gamma_{1,m}/\Gamma_{2,m}$, it follows
that
\begin{align*}
\frac{\partial^{2}F}{\partial p_{1,m}^{2}} & =\frac{B_{c}W_{2,m}}{\left(1/\Gamma_{2,m}+\Omega_{m}\right){}^{2}}-\frac{B_{c}W_{1,m}}{\left(1/\Gamma_{1,m}+\Omega_{m}\right){}^{2}}\\
 & \!\!\!\!=\frac{B_{c}\Gamma_{1,m}^{2}\Gamma_{2,m}^{2}\left(W_{1,m}-W_{2,m}\right)^{2}}{\left(\Gamma_{2,m}-\Gamma_{1,m}\right)^{2}}\!\left(\frac{1}{W_{2,m}}\!-\frac{1}{W_{1,m}}\right)\!\!<\!\!0,
\end{align*}
indicating that $\Omega_{m}$ is a maximizer. \end{proof}
\begin{rem}
In Proposition \ref{P1 WSR1}, the conditions $1<W_{2,m}\text{/}W_{1,m}<\Gamma_{1,m}/\Gamma_{2,m}$ and
$q_{m}>2\Omega_{m}$ are both to avoid a failure of SIC. Indeed, if
$W_{2,m}/W_{1,m}<1$ or $W_{2,m}/W_{1,m}>\Gamma_{1,m}/\Gamma_{2,m}$, the solution to $\mathcal{OP}_{2,m}^{\textrm{SR}1}$
is $p_{1,m}^{\star}=p_{2,m}^{\star}=q_{m}/2$ , i.e., the NOMA system
is unstable according to Definition \ref{SIC-stable}. SIC may also
fail on channel $m$ if $q_{m}\leq2\Omega_{m}$, which will lead to
$p_{1,m}^{\star}=p_{2,m}^{\star}=q_{m}/2$ too. Therefore, the NOMA
system is SIC-stable on channel $m$ if and only if $1<W_{2,m}\text{/}W_{1,m}<\Gamma_{1,m}/\Gamma_{2,m}$
and $q_{m}>2\Omega_{m}$. For all channels, we have the following
result.
\end{rem}
\begin{cor}
\label{C1 WSR1}For $\mathcal{OP}_{1}^{\textrm{SR}1}$, the NOMA system
is SIC-stable only if $P>2\sum_{m=1}^{M}\Omega_{m}$ and $1<W_{2,m}\text{/}W_{1,m}<\Gamma_{1,m}/\Gamma_{2,m}$
for $m=1,\ldots,M$.
\end{cor}
Next, we further optimize the power budget $q_{m}$ for each channel
$m$. To guarantee that the NOMA system is SIC-stable, it is reasonable
to assume that $q_{m}\geq\varTheta_{m}>2\Omega_{m}$ and $P\geq\sum_{m=1}^{M}\varTheta_{m}$
for some positive $\varTheta_{m}$. Then, from $\mathcal{OP}_{1}^{\textrm{SR}1}$
and $\mathcal{OP}_{2,m}^{\textrm{SR}1}$, the corresponding power
budget optimization problem is given by
\[
\mathcal{OP}_{3}^{\textrm{SR}1}:\begin{array}{cl}
\underset{\boldsymbol{q}}{\max} & \sum_{m=1}^{M}f_{m}^{\textrm{SR}1\star}(q_{m})\\
\textrm{s.t.} & \sum_{m=1}^{M}q_{m}\leq P,\ q_{m}\geq\varTheta_{m},\;\forall m
\end{array}
\]
where $f_{m}^{\textrm{SR}1\star}(q_{m})$ is the optimal objective
value of $\mathcal{OP}_{2,m}^{\textrm{SR}1}$ and given by
\begin{align}
f_{m}^{\textrm{SR}1\star}(q_{m})= & W_{1,m}B_{c}\log\left(1+\Omega_{m}\Gamma_{1,m}\right)\nonumber \\
 & +W_{2,m}B_{c}\log\left(\frac{q_{m}\Gamma_{2,m}+1}{\Omega_{m}\Gamma_{2,m}+1}\right).\label{f(q)*WSR1}
\end{align}
It is easily seen that $f_{m}^{\textrm{SR}1\star}(q_{m})$ is a concave
function, so $\mathcal{OP}_{3}^{\textrm{SR}1}$ is a convex problem,
whose solution is provided in the following result.
\begin{thm}
\label{T1 WSR1}The optimal solution to $\mathcal{OP}_{3}^{\textrm{SR}1}$
is given by
\begin{equation}
q_{m}^{\star}=\left[\frac{W_{2,m}B_{c}}{\lambda}-\frac{1}{\Gamma_{2,m}}\right]_{\varTheta_{m}}^{\infty}\label{optimal q_m WSR1}
\end{equation}
where $\lambda$ is chosen such that $\sum_{m=1}^{M}q_{m}^{\star}=P$.
\end{thm}
\begin{proof} The solution of $\mathcal{OP}_{3}^{\textrm{SR}1}$
is given by the well-known waterfilling form. \end{proof} Consequently,
the optimal power allocation for the sum rate maximization with weights
in NOMA systems is jointly characterized by Theorem \ref{T1 WSR1}
and Proposition \ref{P1 WSR1} under the SIC-stability.

\subsection{\label{subsec:WSR2}SR Maximization with QoS (SR2)}

Now, we consider maximizing the SR with QoS constraints. According
to \eqref{max QSR}, in this case the power allocation problem is
given by
\[
\mathcal{OP}_{1}^{\textrm{SR2}}\!\!:\!\!\begin{array}{cl}
\underset{\boldsymbol{p}_{1},\boldsymbol{p}_{2}}{\max}\!\!\!\! & \sum_{m=1}^{M}\left(R_{1,m}(p_{1,m},p_{2,m})+R_{2,m}(p_{1,m},p_{2,m})\right)\\
\textrm{s.t.} & \boldsymbol{0}\leq\boldsymbol{p}_{1}\leq\boldsymbol{p}_{2},\ \sum_{m=1}^{M}(p_{1,m}+p_{2,m})\leq P,\\
 & R_{n,m}\geq R_{n,m}^{\min},\ n=1,2,\ m=1,\ldots,M.
\end{array}
\]
As a special case of $\mathcal{OP}_{1}^{\textrm{SR2}}$, \cite{wang2016power}
studied the power allocation for one channel. Thus, $\mathcal{OP}_{1}^{\textrm{SR2}}$
is still an open problem and its optimal solution is unknown yet.

We use the similar method to address $\mathcal{OP}_{1}^{\textrm{SR2}}$.
By introducing the power budget $q_{m}$ on each channel $m$, $\mathcal{OP}_{1}^{\textrm{SR2}}$
decomposes into the following subproblems for each channel $m$:
\[
\mathcal{OP}_{2,m}^{\textrm{SR2}}:\begin{array}{cl}
\underset{p_{1,m},p_{2,m}}{\max} & R_{1,m}(p_{1,m},p_{2,m})+R_{2,m}(p_{1,m},p_{2,m})\\
\textrm{s.t.} & 0\leq p_{1,m}\leq p_{2,m},\ p_{1,m}+p_{2,m}=q_{m},\\
 & R_{1,m}\geq R_{1,m}^{\min},\ R_{2,m}\geq R_{2,m}^{\min}.
\end{array}
\]
The optimal solution to $\mathcal{OP}_{2,m}^{\textrm{SR2}}$, although
it is nonconvex, is provided in the following result.
\begin{prop}
\label{P1 SR}Suppose that $\Gamma_{1,m}\geq\Gamma_{2,m}$, $A_{2,m}\geq2$,
and $q_{m}\geq\varUpsilon_{m}$, with
\[
\begin{array}{l}
A_{l,m}=2^{\frac{R_{l,m}^{\min}}{B_{c}}},\quad\varUpsilon_{m}\triangleq\frac{A_{2,m}(A_{1,m}-1)}{\Gamma_{1,m}}+\frac{A_{2,m}-1}{\Gamma_{2,m}},\\
\Xi_{m}\triangleq\frac{\Gamma_{2,m}q_{m}-A_{2,m}+1}{A_{2,m}\Gamma_{2,m}}.
\end{array}
\]
Then, the optimal solution to $\mathcal{OP}_{2,m}^{\textrm{SR2}}$
is given by $p_{1,m}^{\star}=\Xi_{m}$ and $p_{2,m}^{\star}=q_{m}-p_{1,m}^{\star}$.
\end{prop}
\begin{proof} Since $p_{2,m}=q_{m}-p_{1,m}$, $p_{1,m}\leq p_{2,m}$
is equal to $p_{1,m}\leq q_{m}/2$ and the objective becomes
\[
T(p_{1,m})\triangleq B_{c}\log\left(1+p_{1,m}\Gamma_{1,m}\right)+B_{c}\log\left(\frac{q_{m}\Gamma_{2,m}+1}{p_{1,m}\Gamma_{2,m}+1}\right).
\]
Given $\Gamma_{1,m}\geq\Gamma_{2,m}$, we take the derivative of $T(p_{1,m})$
and have
\[
\frac{dT}{dp_{1,m}}=\frac{B_{c}}{1/\Gamma_{1,m}+p_{1,m}}-\frac{B_{c}}{1/\Gamma_{2,m}+p_{1,m}}\geq0,
\]
implying that $T(p_{1,m})$ is monotonically nondecreasing, so the
maximum is achieved at the upper bound of $p_{1,m}$. From $R_{1,m}\geq R_{1,m}^{\min}$
and $R_{2,m}\geq R_{2,m}^{\min}$, we obtain
\[
\frac{A_{1,m}-1}{\Gamma_{1,m}}\leq p_{1,m}\leq\Xi_{m}
\]
which holds if and only if $(A_{1,m}-1)/\Gamma_{1,m}\leq\Xi_{m}$,
i.e., $q_{m}\geq\varUpsilon_{m}$. Finally, since $A_{2,m}\geq2$,
$p_{1,m}=\Xi_{m}<q_{m}/2$ holds. Thus the optimal solution is $p_{1,m}^{\star}=\Xi_{m}$.
\end{proof}
\begin{rem}
Similarly, in Proposition \ref{P1 SR}, the conditions $A_{2,m}\geq2$
and $q_{m}\geq\varUpsilon_{m}$ are to guarantee the SIC-stability.
Indeed, if $A_{2,m}<2$, then $\Xi_{m}>q_{m}/2$ and the optimal solution
will be $p_{1,m}^{\star}=p_{2,m}^{\star}=q_{m}/2$, which may leads
a failure of SIC. At the same time, SIC may also fail on channel $m$
if $q_{m}<\varUpsilon_{m}$, which will lead to $p_{1,m}^{\star}=p_{2,m}^{\star}=q_{m}/2$
as well. Therefore, the NOMA system is SIC-stable on channel $m$
if and only if $A_{2,m}\geq2$ and $q_{m}\geq\varUpsilon_{m}$. For
all the channels, we have the following result.
\end{rem}
\begin{cor}
\label{C1 WSR2}For $\mathcal{OP}_{1}^{\textrm{SR2}}$, the NOMA system
is SIC-stable only if $P\geq\sum_{m=1}^{M}\varUpsilon_{m}$ and $A_{2,m}\geq2$
for $m=1,\ldots,M$.
\end{cor}
\begin{rem}
According to Proposition \ref{P1 SR}, if the NOMA system is SIC-stable,
the optimal solution will be $p_{1,m}^{\star}=\Xi_{m}$ and $p_{2,m}^{\star}=q_{m}-p_{1,m}^{\star}$.
Hence, we have $R_{2,m}(p_{1,m}^{\star},p_{2,m}^{\star})=R_{2,m}^{\min}$,
implying that the user with a lower CNR (i.e., $\textrm{U}\textrm{E}_{2,m}$)
receives the power to meet its QoS requirement exactly, while the
remaining power is used to maximize the rate of the user with a higher
CNR (i.e., $\textrm{U}\textrm{E}_{1,m}$).
\end{rem}
Then, we focus on optimizing the power budget $q_{m}$ for each channel.
Similarly, to guarantee the NOMA system is SIC-stable, we assume that
$q_{m}\geq\varUpsilon_{m}$ and $P\geq\sum_{m=1}^{M}\varUpsilon_{m}$.
According to $\mathcal{OP}_{1}^{\textrm{SR}2}$ and $\mathcal{OP}_{2,m}^{\textrm{SR2}}$,
the corresponding power budget optimization problem is as follows
\[
\mathcal{OP}_{3}^{\textrm{SR2}}:\begin{array}{c}
\underset{\boldsymbol{q}}{\max}\\
\textrm{s.t.}
\end{array}\begin{array}{l}
\sum_{m=1}^{M}f_{m}^{\textrm{SR2}\star}(q_{m})\\
\sum_{m=1}^{M}q_{m}\leq P,\ q_{m}\geq\varUpsilon_{m},\;\forall m
\end{array}
\]
where $f_{m}^{\textrm{SR2}\star}(q_{m})$ is the optimal objective
value of $\mathcal{OP}_{2,m}^{\textrm{SR2}}$ and given by
\begin{equation}
f_{m}^{\textrm{SR2}\star}(q_{m})=w(q_{m})+R_{2,m}^{\min}\label{f*SR2}
\end{equation}
where $w(q_{m})=B_{c}\log\frac{\left(A_{2,m}\Gamma_{2,m}-A_{2,m}\Gamma_{1,m}+\Gamma_{1,m}\Gamma_{2,m}q_{m}+\Gamma_{1,m}\right)}{A_{2,m}\Gamma_{2,m}}$.
Since $f_{m}^{\textrm{SR2}\star}(q_{m})$ is a concave function, $\mathcal{OP}_{3}^{\textrm{SR2}}$
is a convex problem, whose solution is also given in a waterfilling
form.
\begin{thm}
\label{T1 WSR2}The optimal solution to $\mathcal{OP}_{3}^{\textrm{SR2}}$
is given by
\[
q_{m}^{\star}=\left[\frac{B_{c}}{\textrm{\ensuremath{\lambda}}}-\frac{A_{2,m}}{\Gamma_{1,m}}+\frac{A_{2,m}}{\Gamma_{2,m}}-\frac{1}{\Gamma_{2,m}}\right]_{\varUpsilon_{m}}^{\infty}
\]
where $\lambda$ is chosen such that $\sum_{m=1}^{M}q_{m}^{\star}=P$.
\end{thm}
\begin{proof} The proof is simple and thus omitted. \end{proof}
Therefore, the optimal power allocation for the SR maximization with
QoS constraints in NOMA systems is jointly characterized by Proposition
\ref{P1 SR} and Theorem \ref{T1 WSR2}. Note that, unlike the MMF
criterion, for the SR maximization with weights or QoS constraints,
NOMA systems are not always SIC-stable but have to satisfy some conditions
on the weights, power budgets, and QoS thresholds as indicated in
this section.

\section{\label{sec:GEE}Optimal Power Allocation for Energy Efficiency}

In this section, we investigate the optimal power allocation for maximizing
the energy efficiency (EE) of the NOMA systems with weights or QoS
constraints.

\subsection{EE Maximization with Weights (EE1)}

According to (\ref{max WEE}), with given channel assignment, the
problem of maximizing the EE with weights is equivalent to the following
power allocation problem:
\[
\mathcal{OP}_{1}^{\textrm{EE}1}:\begin{array}{ll}
\underset{\boldsymbol{p_{1},p_{2}}}{\max} & \frac{\sum_{m=1}^{M}g(p_{1,m},p_{2,m})}{P_{T}+\sum_{m=1}^{M}\left(p_{1,m}+p_{2,m}\right)}\\
\textrm{s.t.} & \boldsymbol{0}\leq\boldsymbol{p}_{1}\leq\boldsymbol{p}_{2},\ \sum_{m=1}^{M}\left(p_{1,m}+p_{2,m}\right)\leq P
\end{array}
\]
where $g(p_{1,m},p_{2,m})\triangleq W_{1,m}R_{1,m}(p_{1,m},p_{2,m})+W_{2,m}R_{2,m}(p_{1,m},p_{2,m}).$
The difficulties in solving $\mathcal{OP}_{1}^{\textrm{EE}1}$ lie
in its nonconvex and fractional objective. In the literature, only
\cite{zhang2016energy,fang2016energy} investigated this problem,
whereas \cite{zhang2016energy} only found the optimal solution in
the special case $M=1$, i.e., a single channel, and \cite{fang2016energy}
obtained a suboptimal power allocation solution. In the following,
we will show that this problem can also be optimally solved.

We use the similar trick to address this problem, i.e., introducing
the auxiliary variables $\left\{ q_{m}\right\} _{m=1}^{M}$ with $p_{1,m}+p_{2,m}=q_{m}$
for each channel $m$. Then, $\mathcal{OP}_{1}^{\textrm{EE}1}$ is
decomposed into the following subproblems for each channel $m$:
\[
\mathcal{OP}_{2,m}^{\textrm{EE}1}:\begin{array}{cl}
\begin{array}{c}
\underset{p_{1,m},p_{2,m}}{\max}\end{array} & \frac{g(p_{1,m},p_{2,m})}{P_{T}+\sum_{k=1}^{M}q_{k}}\\
\textrm{s.t.} & 0\leq p_{1,m}\leq p_{2,m},\ p_{1,m}+p_{2,m}=q_{m}
\end{array}
\]
whose optimal solution is provided in the following.
\begin{prop}
\label{P1-EE1}Suppose that $\Gamma_{1,m}\geq\Gamma_{2,m}$, $1<W_{2,m}\text{/}W_{1,m}<\Gamma_{1,m}/\Gamma_{2,m}$
and $q_{m}>2\Omega_{m}$, with
\[
\Omega_{m}\triangleq\frac{W_{2,m}\Gamma_{2,m}-W_{1,m}\Gamma_{1,m}}{\Gamma_{1,m}\Gamma_{2,m}\left(W_{1,m}-W_{2,m}\right)}.
\]
Then the optimal solution to $\mathcal{OP}_{2,m}^{\textrm{EE}1}$
is same with $\mathcal{OP}_{2,m}^{\textrm{SR}1}$, i.e., $p_{1,m}^{\star}=\Omega_{m}$
and $p_{2,m}^{\star}=q_{m}-p_{1,m}^{\star}$.
\end{prop}
\begin{rem}
It is not difficult to see that with given channel power budgets $\{q_{m}\}_{m=1}^{M}$,
$\mathcal{OP}_{2,m}^{\textrm{EE}1}$ is actually equivalent to $\mathcal{OP}_{2,m}^{\textrm{SR}1}$,
so they have the same optimal solution. Therefore, we obtain the same
SIC-stability conditions for $\mathcal{OP}_{2,m}^{\textrm{EE}1}$:
the NOMA system is SIC-stable on channel $m$ if and only if $q_{m}>2\Omega_{m}$
and $1<W_{2,m}\text{/}W_{1,m}<\Gamma_{1,m}/\Gamma_{2,m}$, and is SIC-stable on all channels only if
$P>2\sum_{m=1}^{M}\Omega_{m}$ and $1<W_{2,m}\text{/}W_{1,m}<\Gamma_{1,m}/\Gamma_{2,m}$ for $m=1,\ldots,M$.
\end{rem}
Then, we concentrate on searching the optimal power budget $q_{m}$
for each channel. Similarly, to guarantee the NOMA system is SIC-stable,
it is assumed that $q_{m}\geq\varTheta_{m}>2\Omega_{m}$ and $P\geq\sum_{m=1}^{M}\varTheta_{m}$
for some positive $\varTheta_{m}$. According to Proposition \ref{P1-EE1},
$\mathcal{OP}_{1}^{\textrm{EE}1}$, and $\mathcal{OP}_{2,m}^{\textrm{EE}1}$,
the power budget optimization problem is formulated as
\[
\mathcal{OP}_{3}^{\textrm{EE}1}:\begin{array}{cl}
\underset{\boldsymbol{q}}{\max} & \eta(\boldsymbol{q})\triangleq\frac{\sum_{m=1}^{M}f_{m}^{\textrm{SR}1\star}(q_{m})}{P_{T}+\sum_{m=1}^{M}q_{m}}\\
\textrm{s.t.} & \sum_{m=1}^{M}q_{m}\leq P,\ q_{m}\geq\varTheta_{m},\;\forall m
\end{array}
\]
where $f_{m}^{\textrm{SR}1\star}(q_{m})$ is the optimal value of
$\mathcal{OP}_{2,m}^{\textrm{SR}1}$ and given in \eqref{f(q)*WSR1}.
Although $f_{m}^{\textrm{SR}1\star}(q_{m})$ is a concave function,
$\mathcal{OP}_{3}^{\textrm{EE}1}$ is nonconvex due to the fraction
form. To solve it, we introduce the following objective function:
\begin{align*}
H(\boldsymbol{q},\alpha)\triangleq & \sum_{m=1}^{M}f_{m}^{\textrm{SR}1\star}(q_{m})-\alpha\left(P_{T}+\sum_{m=1}^{M}q_{m}\right)\\
= & \sum_{m=1}^{M}\left(\tilde{R}_{1,m}+W_{2,m}B_{c}\log\left(\frac{q_{m}\Gamma_{2,m}+1}{\Omega_{m}\Gamma_{2,m}+1}\right)\right)\\
 & -\alpha\left(P_{T}+\sum_{m=1}^{M}q_{m}\right)
\end{align*}
where $\tilde{R}_{1,m}\triangleq W_{1,m}B_{c}\log\left(1+\Omega_{m}\Gamma_{1,m}\right)$
and $\alpha$ is a positive parameter. Then, we consider the following
convex problem with given $\alpha$:
\[
\mathcal{OP}_{4}^{\textrm{EE}1}:\begin{array}{cl}
\underset{\boldsymbol{q}}{\max} & H\left(\boldsymbol{q},\alpha\right)\\
\textrm{s.t.} & \sum_{m=1}^{M}q_{m}\leq P,\ q_{m}\geq\varTheta_{m},\;\forall m.
\end{array}
\]
The relation between $\mathcal{O}\mathcal{P}_{3}^{\textrm{EE}1}$
and $\mathcal{O}\mathcal{P}_{4}^{\textrm{EE}1}$ is given by the following
lemma.
\begin{lem}
\label{EE1 Lemma1}(\cite[pp. 493-494]{dinkelbach1967nonlinear})
Let $H^{\star}\left(\alpha\right)$ be the optimal objective value
of $\mathcal{O}\mathcal{P}_{4}^{\textrm{EE}1}$ and $\boldsymbol{q}^{\star}(\alpha)$
be the optimal solution of $\mathcal{O}\mathcal{P}_{4}^{\textrm{EE}1}$.
Then, $\boldsymbol{q}^{\star}(\alpha)$ is the optimal solution to
$\mathcal{OP}_{3}^{\textrm{EE}1}$ if and only if $H^{\star}\left(\alpha\right)=0$.
\end{lem}
Lemma \ref{EE1 Lemma1} indicates the optimal solution to $\mathcal{O}\mathcal{P}_{3}^{\textrm{EE}1}$
can be found by solving $\mathcal{O}\mathcal{P}_{4}^{\textrm{EE}1}$
parameterized by $\alpha$ and then updating $\alpha$ until $H^{\star}\left(\alpha\right)=0$.
For this purpose, we first solve $\mathcal{O}\mathcal{P}_{4}^{\textrm{EE}1}$
with given $\alpha$, whose solution is provided in the following
result.
\begin{thm}
\label{T1-GEE1}The optimal solution to $\mathcal{OP}_{4}^{\textrm{EE}1}$
is
\begin{equation}
q_{m}^{\star}=\left[\frac{W_{2,m}B_{c}}{\alpha+\lambda}-\frac{1}{\Gamma_{2,m}}\right]_{\varTheta_{m}}^{\infty}\label{optimal q_m GEE1}
\end{equation}
where $\lambda$ is chosen such that $\sum_{m=1}^{M}q_{m}^{\star}=P$
.
\end{thm}
\begin{proof} The solution is obtained by exploiting the KKT conditions
of $\mathcal{O}\mathcal{P}_{4}^{\textrm{EE}1}$. \end{proof} After
the optimal solution to $\mathcal{OP}_{4}^{\textrm{EE}1}$ is obtained,
we shall find an $\alpha$ such that $H^{\star}\left(\alpha\right)=0$.
This can be achieved by Algorithm \ref{EE Algorithm}, which is guaranteed
to converge to the desirable $\alpha$ \cite{dinkelbach1967nonlinear}.
Thereby, the optimal power allocation for the EE maximization with
weights in NOMA systems is provided by Algorithm \ref{EE Algorithm},
Proposition \ref{P1-EE1} and Theorem \ref{T1-GEE1}.

\begin{algorithm}
\protect\protect\caption{\label{EE Algorithm}Channel Power Budget Optimization for EE}

1:\textbf{ Initialization}: set $\alpha_{ini}=0$, $H_{ini}^{\star}=\infty$
and precision $\delta>0$.

2:\textbf{ While} $\left|H^{\star}\left(\alpha\right)\right|>\delta$
\textbf{do}

3:~~~~~~Find the optimal $\boldsymbol{q}^{\star}$ according
to Theorem \ref{T1-GEE1};

4:~~~~~~Calculate $H^{\star}\left(\alpha\right)$ ;

5:~~~~~~Update $\alpha=\eta\left(\boldsymbol{q}^{\star}\right)$;

6:.\textbf{Return} $\alpha$ and $\boldsymbol{q}^{\star}$.
\end{algorithm}

\subsection{EE Maximization with QoS (EE2)}

In this subsection, we consider maximizing the EE with QoS constraints.
According to \eqref{max QEE}, in this case the power allocation problem
is given by
\[
\mathcal{OP}_{1}^{\textrm{EE}2}:\begin{array}{rl}
\underset{\boldsymbol{p}_{1},\boldsymbol{p}_{2}}{\max} & \frac{\sum_{m=1}^{M}\left(R_{1,m}\left(p_{1,m},p_{2,m}\right)+R_{2,m}\left(p_{1,m},p_{2,m}\right)\right)}{P_{T}+\sum_{m=1}^{M}\left(p_{1,m}+p_{2,m}\right)}\\
\textrm{s.t. } & \boldsymbol{0}\leq\boldsymbol{p}_{1}\leq\boldsymbol{p}_{2},\ \sum_{m=1}^{M}(p_{1,m}+p_{2,m})\leq P,\\
 & R_{n,m}\geq R_{n,m}^{\min},\ n=1,2,\ m=1,\ldots,M.
\end{array}
\]
Similarly, $\mathcal{OP}_{1}^{\textrm{EE}2}$ is a nonconvex fractional
optimization problem. In the literature, the EE maximization with
QoS constraints has only been studied in \cite{zhang2016energy},
but the optimal solution was only found for one channel, i.e., $M=1$.
Therefore, $\mathcal{OP}_{1}^{\textrm{EE}2}$ is an open problem and
its solution in the general case is still unknown.

To solve $\mathcal{OP}_{1}^{\textrm{EE}2}$, we also adopt $\left\{ q_{m}\right\} _{m=1}^{M}$
with $p_{1,m}+p_{2,m}=q_{m}$ and decompose $\mathcal{OP}_{1}^{\textrm{EE}2}$
into a group of subproblems for each channel $m$:
\begin{align*}
\underset{p_{1,m},p_{2,m}}{\max} & \frac{R_{1,m}\left(p_{1,m},p_{2,m}\right)+R_{2,m}\left(p_{1,m},p_{2,m}\right)}{P_{T}+\sum_{k=1}^{M}q_{k}}\\
\mathcal{OP}_{2,m}^{\textrm{EE2}}:\qquad\textrm{s.t.\ } & \ 0\leq p_{1,m}\leq p_{2,m},\ p_{1,m}+p_{2,m}=q_{m},\\
 & \ R_{1,m}\geq R_{1,m}^{\min},\ R_{2,m}\geq R_{2,m}^{\min}
\end{align*}
whose solution is given in the following closed form.
\begin{prop}
\label{P1-EE2}Suppose that $\Gamma_{1,m}\geq\Gamma_{2,m}$, $A_{2,m}\geq2$,
and $q_{m}\geq\varUpsilon_{m}$, with
\[
\begin{array}{l}
A_{l,m}=2^{\frac{R_{l,m}^{\min}}{B_{c}}},\quad\varUpsilon_{m}\triangleq\frac{A_{2,m}(A_{1,m}-1)}{\Gamma_{1,m}}+\frac{A_{2,m}-1}{\Gamma_{2,m}},\\
\Xi_{m}\triangleq\frac{\Gamma_{2,m}q_{m}-A_{2,m}+1}{A_{2,m}\Gamma_{2,m}}.
\end{array}
\]
Then, the optimal solution to $\mathcal{OP}_{2,m}^{\textrm{EE}2}$
is $p_{1,m}^{\star}=\Xi_{m}$ and $p_{2,m}^{\star}=q_{m}-p_{1,m}^{\star}$.
\end{prop}
\begin{rem}
It is not surprising that the solutions of $\mathcal{OP}_{2,m}^{\textrm{EE2}}$
and $\mathcal{OP}_{2,m}^{\textrm{SR2}}$ coincide, since $\mathcal{OP}_{2,m}^{\textrm{EE2}}$
is equivalent to $\mathcal{OP}_{2,m}^{\textrm{SR2}}$ with given $\left\{ q_{m}\right\} _{m=1}^{M}$.
Therefore, the same SIC-stability conditions hold for $\mathcal{OP}_{2,m}^{\textrm{EE}2}$
and $\mathcal{OP}_{1}^{\textrm{EE}2}$: the NOMA system is SIC-stable
on channel $m$ if and only if $q_{m}\geq\varUpsilon_{m}$ and $A_{2,m}\geq2$,
and is SIC-stable on all channels only if $P\geq\sum_{m=1}^{M}\varUpsilon_{m}$
and $A_{2,m}\geq2$ for $m=1,\ldots,M$.
\end{rem}
Next, we optimize the channel power budget $q_{m}$ for each channel.
First, we assume that $q_{m}\geq\varUpsilon_{m}$ and $P\geq\sum_{m=1}^{M}\varUpsilon_{m}$
to guarantee the SIC-stability. Then, according to Proposition \ref{P1-EE2},
$\mathcal{OP}_{1}^{\textrm{EE}2}$, and $\mathcal{OP}_{2,m}^{\textrm{EE}2}$,
the power budget optimization problem is given by
\[
\mathcal{OP}_{3}^{\textrm{EE}2}:\begin{array}{cl}
\underset{\boldsymbol{q}}{\max} & \eta(\boldsymbol{q})\triangleq\frac{\sum_{m=1}^{M}f_{m}^{\textrm{SR}2\star}(q_{m})}{P_{T}+\sum_{m=1}^{M}q_{m}}\\
\textrm{s.t. } & \sum_{m=1}^{M}q_{m}\leq P,\ q_{m}\geq\Upsilon_{m},\;\forall m
\end{array}
\]
where $f_{m}^{\textrm{SR}2\star}(q_{m})$ is the optimal value of
$\mathcal{OP}_{2,m}^{\textrm{SR}2}$ and given in (\ref{f*SR2}).
To solve $\mathcal{OP}_{3}^{\textrm{EE}2}$, we also introduce a parameterized
objective function
\begin{align*}
Q(\boldsymbol{q},\alpha) & \triangleq\sum_{m=1}^{M}f_{m}^{\textrm{SR}2\star}(q_{m})-\alpha\left(P_{T}+\sum_{m=1}^{M}q_{m}\right)\\
 & =\sum_{m=1}^{M}\left(w(q_{m})+R_{2,m}^{\min}\right)-\alpha\left(P_{T}+\sum_{m=1}^{M}q_{m}\right)
\end{align*}
where $w(q_{m})=B_{c}\log\frac{\left(A_{2,m}\Gamma_{2,m}-A_{2,m}\Gamma_{1,m}+\Gamma_{1,m}\Gamma_{2,m}q_{m}+\Gamma_{1,m}\right)}{A_{2,m}\Gamma_{2,m}}$,
$\alpha$ is a positive parameter, and formulate the following problem
with given $\alpha$:
\[
\mathcal{OP}_{4}^{\textrm{EE}2}:\begin{array}{cl}
\underset{\boldsymbol{q}}{\max} & Q\left(\boldsymbol{q},\alpha\right)\\
\textrm{s.t.} & \sum_{m=1}^{M}q_{m}\leq P,\ q_{m}\geq\Upsilon_{m},\;\forall m.
\end{array}
\]
Then, according to Lemma \ref{EE1 Lemma1}, the optimal solution to
$\mathcal{O}\mathcal{P}_{3}^{\textrm{EE}2}$ can be found by solving
$\mathcal{O}\mathcal{P}_{4}^{\textrm{EE}2}$ for a given $\alpha$
and then updating $\alpha$ until the optimal objective value of $\mathcal{O}\mathcal{P}_{4}^{\textrm{EE}2}$,
denoted by $Q^{\star}(\alpha)$, satisfies $Q^{\star}(\alpha)=0$.
Therefore, we first solve $\mathcal{O}\mathcal{P}_{4}^{\textrm{EE}2}$,
which is a convex problem since $Q(\boldsymbol{q},\alpha)$ is concave
in $\boldsymbol{q}$. In particular, the optimal solution to $\mathcal{O}\mathcal{P}_{4}^{\textrm{EE}2}$
is provided as follows.
\begin{thm}
\label{T1-GEE2} The optimal solution to $\mathcal{OP}_{4}^{\textrm{EE}2}$
is
\[
q_{m}^{\star}=\left[\frac{W_{1,m}B_{c}}{\lambda+\alpha}-\frac{A_{2,m}}{\Gamma_{1,m}}+\frac{A_{2,m}}{\Gamma_{2,m}}-\frac{1}{\Gamma_{2,m}}\right]_{\Upsilon_{m}}^{\infty}
\]
where $\lambda$ is chosen such that $\sum_{m=1}^{M}q_{m}^{\star}=P$
.
\end{thm}
\begin{proof} The solution is obtained by exploiting the KKT conditions
of $\mathcal{O}\mathcal{P}_{4}^{\textrm{EE}2}$. \end{proof} After
obtaining the optimal solution to $\mathcal{OP}_{4}^{\textrm{GEE}2}$,
we shall find an $\alpha$ such that $Q^{\star}(\alpha)=0$. This
can also be achieved by Algorithm \ref{EE Algorithm}, where Theorem
\ref{T1-GEE1} and $H^{\star}\left(\alpha\right)$ are replaced by
Theorem \ref{T1-GEE2} and $Q^{\star}(\alpha)$, respectively. Consequently,
the optimal power allocation for the EE maximization with QoS constraints
in NOMA systems is obtained by using Theorem \ref{T1-GEE2}, Proposition
\ref{P1-EE2}, and Algorithm \ref{EE Algorithm}.

\section{\label{sec:MATCHING}Channel Assignment}

In the previous sections, we have found the optimal power allocation
with given channel assignment under various performance criteria for
NOMA systems. Specifically, we first achieve the optimal power allocation
of the users on each channel, which can be expressed as functions
of the power budget of each channel. Then, we further optimized the
power budgets of all channels and thus obtained the optimal multichannel
power allocation, which can be characterized in a closed or semi-closed
form. In this section, we consider the joint optimization of power
allocation and channel assignment in NOMA systems. Unfortunately,
such a joint optimization problem has been shown to be NP-hard \cite{lei2015joint}.
Hence, finding the jointly optimal solution generally requires exhaustive
search \cite{zhang2016radio}, which results in prohibitive computational
complexity. Therefore, in practice suboptimal but efficient joint
optimization methods are more preferred \cite{fang2016energy,parida2014power,zhang2016radio,hojeij2015resource}.

Enlightened by the optimal multichannel power allocation obtained
above, in this paper we propose a low-complexity method to jointly
optimize the power allocation and channel assignment in NOMA systems.
Specifically, we incorporate the dynamic matching algorithm\cite{2014dynamicmatching},
which is an efficient method to deal with assignment problems, with
our optimal power allocation, and iteratively exploit them to refine
the solution.

To describe the dynamic matching between the users and the channels,
we consider channel assignment as a two-sided matching problem between
the set of $N$ users and the set of $M$ channels, where $N=2M$
since each channel is shared by two users. Denote channel $m$ by
$C_{m}$. We say $\textrm{U}\textrm{E}_{n}$ and $C_{m}$ are matched
with each other if $\textrm{U}\textrm{E}_{n}$ is assigned on $C_{m}$.
Moreover, denote $PF\left(\textrm{U}\textrm{E}_{n}\right)$ for $n=1,2,\cdots N$
and $PF\left(C_{m}\right)$ for $m=1,2,\cdots M$ to be the preference
lists of the users and channels, respectively. We say $\textrm{U}\textrm{E}_{n}$
prefers $C_{i}$ to $C_{j}$ if $\textrm{U}\textrm{E}_{n}$ has a
higher channel gain on $C_{i}$ than on $C_{j}$, and it can be expressed
as
\[
C_{i}(n)\succ C_{j}(n).
\]
In addition, we say that $C_{m}$ prefers user set $\varsigma_{l}$
to user set $\varsigma_{k}$ (where $\varsigma_{l}$ and $\varsigma_{k}$
are the subsets of $\left\{ 1,2,\cdots N\right\} $) if the users
in set $\varsigma_{l}$ can provide better performance than the users
in set $\varsigma_{k}$ on $C_{m}$. This preference is expressed
as
\[
O_{m}(\varsigma_{l})>O_{m}(\varsigma_{k}),\ \varsigma_{l},\varsigma_{k}\subset\left\{ \textrm{U}\textrm{E}_{1},\textrm{U}\textrm{E}_{2}\cdots\textrm{U}\textrm{E}_{N_{m}}\right\}
\]
where $O_{m}(\varsigma_{l})$ denotes the performance measure of user
set $\varsigma_{l}$ on $C_{l}$, which could be the maximin fairness,
sum rate, or energy efficiency introduced in Section \ref{subsec:Problem-Formulation}.
Now, we are ready to formulate the channel assignment optimization
as a two-side matching problem according to matching theory \cite{2017two-sidematching}.
The following definition formally introduces a matching in the NOMA
system.
\begin{defn}
\label{def:matching}Consider the users and channels as two disjoint
sets. A two-to-one matching $\Phi$ is a mapping from all the subsets
of users $\boldsymbol{N}$ into the channels set $\boldsymbol{M}$,
which satisfies the following properties for $UE_{n}\in\boldsymbol{N}$
and $C_{m}\in\boldsymbol{M}$

(a) $\Phi\left(\textrm{U}\textrm{E}_{n}\right)\in\boldsymbol{M}$;

(b) $\Phi^{-1}\left(C_{m}\right)\subseteq\boldsymbol{N}$;

(c) $\left|\Phi\left(\textrm{U}\textrm{E}_{n}\right)\right|=1$, $\left|\Phi^{-1}\left(C_{m}\right)\right|=2$;

(d) $C_{m}\in\Phi\left(\textrm{U}\textrm{E}_{n}\right)\iff\textrm{U}\textrm{E}_{n}\in\Phi^{-1}\left(C_{m}\right)$.
\end{defn}
In Definition \ref{def:matching}, property (a) states that each user
matches with one channel, property (b) indicates that each channel
can be matched with a subset of users, property (c) states that each
channel can only be assigned to two users, and property (d) means
that $UE_{n}$ and $C_{m}$ are matched with each other. Consequently,
the channel assignment problem is to identify a matching between the
users and the channels. However, the globally optimal matching that
maximizes the aggregate performance of all users is hard to find and
usually requires exhaustive search. Instead, in practice, people are
more interested in seeking a so-called stable matching, which can
be efficiently found by the deferred acceptance (DA) procedure \cite{2016Scableand}.
\begin{defn}
Given a matching $\Phi$ such that $\textrm{U}\textrm{E}_{n}\notin\Phi^{-1}\left(C_{m}\right)$
and $C_{m}\notin\Phi\left(\textrm{U}\textrm{E}_{n}\right)$. If $O_{m}(S_{\textrm{new}})>O_{m}\left(\Phi^{-1}\left(C_{m}\right)\right)$
where $S_{\textrm{new}}\subseteq\left\{ \textrm{U}\textrm{E}_{n}\right\} \cup S$
and $S=\Phi^{-1}\left(C_{m}\right)$, then $S_{\textrm{new}}$ is
the preferred user set for $C_{m}$ and $\left(\textrm{U}\textrm{E}_{n},C_{m}\right)$
is a preferred pair.
\end{defn}
\begin{algorithm}
\protect\protect\caption{\label{alg:CMA-Matching}Channel Assignment via Matching }

1:\textbf{ Initialize}:

\hspace*{0.5cm}1) $S_{\textrm{Match}}(m)$ is the matched list to
record users matched

\hspace*{0.9cm}on $C_{m}$, $m=\left\{ 1,2\cdots M\right\} $.

\hspace*{0.5cm}2) $S_{\textrm{UnMatch}}$ is the set of unmatched
users.

\hspace*{0.5cm}3) Obtain preference lists $PF\left(\textrm{U}\textrm{E}_{n}\right),n=\left\{ 1,2,\cdots N\right\} $

\hspace*{0.9cm} and $PF\left(C_{m}\right),m=\left\{ 1,2\cdots M\right\} $
according to CNRs.

2:\textbf{ while} $\left\{ S_{\textrm{UnMatch}}\right\} $ is not
empty

3: \hspace*{0.3cm}\textbf{for} $n=1$ to $N$ \textbf{do }

\hspace*{0.8cm}Each user sends matching request to its most preferred

\hspace*{0.8cm}channel $m^{*}$ according to its preference list
$PF\left(\textrm{U}\textrm{E}_{n}\right)$.

4: \hspace*{0.6cm}\textbf{if} $\left|S_{\textrm{Match}}\left(m^{*}\right)\right|<2$
\textbf{then}

\hspace*{1.2cm}Channel $m^{*}$ adds $\textrm{U}\textrm{E}_{n}$
to $S_{\textrm{Match}}\left(m^{*}\right)$ and removes

\hspace*{1.2cm}$\textrm{U}\textrm{E}_{n}$ from $\left\{ S_{\textrm{UnMatch}}\right\} .$

\hspace*{1cm}\textbf{end if}

5: \hspace*{0.6cm} \textbf{if }$\left|S_{\textrm{Match}}\left(m^{*}\right)\right|=2$
\textbf{then}

\hspace*{1.2cm}1) Identify the power allocation for every two users

\hspace*{1.2cm}in $S_{\varsigma_{L}}$,$S_{\varsigma_{L}}\subset\left\{ S_{\textrm{Match}}(m^{*}),n\right\} $
according to the

\hspace*{1.2cm}corresponding proposition with $q_{m^{*}}$.

\hspace*{1.2cm}2) Channel $m^{*}$ selects a set of $2$ users $S_{\varsigma_{L}}$
satisfying

\hspace*{1.2cm}the objective functions $O_{m^{*}}\left(\varsigma_{l}\right)>O_{m^{*}}\left(\varsigma_{k}\right),$

\hspace*{1.2cm}$\varsigma_{l},\varsigma_{k}\subset\left\{ S_{\textrm{Match}}\left(m^{*}\right),n\right\} $.

\noindent \hspace*{1.2cm}3) Channel $m^{*}$ sets $S_{\textrm{Match}}\left(m^{*}\right)=\varsigma_{l}$,and
reject other

\noindent \hspace*{1.2cm} users. Remove the allocated users from
$\left\{ S_{\textrm{UnMatch}}\right\} $,

\noindent \hspace*{1.2cm}add the unallocated user to$\left\{ S_{\textrm{UnMatch}}\right\} $.

\noindent \hspace*{1.2cm}4) The rejected user remove channel $m^{*}$
from their

\noindent \hspace*{1.2cm}preference lists.

\hspace*{1cm} \textbf{end if}

\hspace*{0.7cm}\textbf{end for}

\hspace*{0.2cm}\textbf{end while}
\end{algorithm}

According to the DA procedure, each user sends a matching request
to its most preferred channel according to its preference list, while
this preferred channel has the right to accept or reject the user
according to the performance that all users can achieve on this channel.
Thus, the DA procedure is to find preferred pairs for each user and
each channel, which is formally described in Algorithm \ref{alg:CMA-Matching}.
In Algorithm \ref{alg:CMA-Matching}, the first step is to initialize
the preference lists of channels and users according to the CNRs.
Meantime, $S_{\textrm{Match}}(m)$ and $S_{\textrm{UnMatch}}(m)$
for $m=1,2,\cdots M$ are respectively initialized to record the allocated
users on $C_{m}$ and unallocated users. The next step is the matching
procedure, where at each round each user sends a matching request
according to its preferred list $PF\left(\textrm{U}\textrm{E}_{n}\right)$
for $n=1,2,\cdots N$ . Then, the channel accepts the user directly
if the number of the users on this channel is less than two, otherwise
only the user that can improve the performance will be accepted. This
matching process will terminate when there is no user left to be matched.

\begin{algorithm}
\caption{\label{Joint-Channel-Assignment}Joint Channel Assignment and Power
Allocation Optimization}

1: \textbf{Initialize}: $q_{m}=\frac{P}{M}$ for $m=1,\ldots,M$;

2: \textbf{Repeat}

3:\hspace*{1.2cm}Obtain channel assignment $\left\{ S_{\textrm{ Match}}\left(m\right)\right\} _{m=1}^{M}$
using

\hspace*{1.45cm}Algorithm \ref{alg:CMA-Matching};

4:\hspace*{1.2cm}Compute $\boldsymbol{p}_{1}^{\star}$, $\boldsymbol{p}_{2}^{\star}$,
$\boldsymbol{q}^{\star}$ according to the results in this

\hspace*{1.45cm}paper;

5: \textbf{Until }the prescribed iteration number is reached;

6: \textbf{Output}: $\left\{ S_{\textrm{ Match}}\left(m\right)\right\} _{m=1}^{M}$,
$\boldsymbol{p}_{1}^{\star}$, $\boldsymbol{p}_{2}^{\star}$.
\end{algorithm}

Now, we are able to jointly optimize channel assignment and power
allocation by using Algorithm \ref{alg:CMA-Matching} and the optimal
power allocation obtained in this paper, which is described in Algorithm
\ref{Joint-Channel-Assignment}. In the initialization, the BS allocates
equal power budgets to all channels. In the next step, we obtain the
channel assignment using Algorithm \ref{alg:CMA-Matching}, then update
the optimal power allocation for each user and power budget for each
channel, and so on.
\begin{rem}
\textcolor{black}{It is worth pointing out that the optimal power
allocation provided in this paper can be jointly used not only with
the DA matching algorithm (i.e., Algorithm \ref{alg:CMA-Matching})
but also with any other assignment algorithms. One just needs to replace
Algorithm \ref{alg:CMA-Matching} by the desirable assignment algorithm
in Algorithm \ref{Joint-Channel-Assignment}. Furthermore, our results
can also reduce the complexity of exhaustive search. Indeed, in \cite{sun2016optimal},
the exhaustive search was performed in the joint continuous-discrete
dimension of powers and channels. Now, given the optimal power allocation
of all users over multiple channels for fixed channel assignment,
one can focus on searching the optimal channel assignment by, e.g.,
checking all possible user-channel matchings, which is a pure combinatorial
problem but not a mixed one anymore. In fact, we will show in the
next section that the performance of the proposed low-complexity joint
resource optimization method is quite close to that of the globally
optimal solution found by exhaustive search.}
\end{rem}

\section{\label{sec:SIMULATION}Numerical Results}

\textcolor{black}{In this section, we evaluate the performance of
the optimal power allocation and the proposed joint resource optimization
method via numerical simulations. In simulations, the base station
is located in the cell center and the users are randomly distributed
in a circular range with a radius of $300$m. The minimum distance
between users is set to be $30$m, and the minimum distance between
users and BS is $40$m. Each channel coefficient follows an i.i.d.
Gaussian distribution as $g_{m}\sim\mathcal{CN}(0,1)$ for $m=1,\ldots,M$
and the path loss exponent is $\alpha=2$. The total power budget
of the BS is $P=41\textrm{dBm}$ and the circuit power consumption
is $P_{T}=30\textrm{dBm}$. The noise power is $\sigma_{m}^{2}=BN_{0}/M$,
where the bandwidth is $B=5\textrm{MHz}$ and the noise power spectral
density is $N_{0}=-174\textrm{dBm}$. We set the user weights to be
$W_{1,m}=0.9$ and $W_{2,m}=1.1$ for $\forall m$ and the QoS thresholds
to be $R_{l,m}^{\min}=2$ bps/Hz for $l=1,2$, $\forall m$. We compare
the proposed joint resource allocation (JRA) method that uses the
optimal power allocation and matching algorithm, with OFDMA where
each user occupies a bandwidth $B/N$ and the power is optimized in
a waterfilling manner, the DC method used in \cite{fang2016energy,parida2014power}
where the power allocation was optimized via DC programming, the conventional
user pairing (CUP) method used in \cite{7511620} where the same channel
is assigned to the users with a significant channel gain difference,
and the exhaustive search.}

\textcolor{black}{Fig. 1 depicts the minimum user rates of the NOMA
system using the proposed JRA method under the MMF criterion (NOMA
JRA), the CUP method with our optimal power allocation for MMF (NOMA
CUP), and the OFDMA system for different total power budgets and user
numbers. It is clearly seen that NOMA is better than OFDMA in terms
of user fairness and the performance gap between NOMA and OFDMA becomes
larger as the number of users increases. Meantime, the proposed JRA
method outperforms the CUP method.}

\begin{figure}[h]
\centering \includegraphics[scale=0.6]{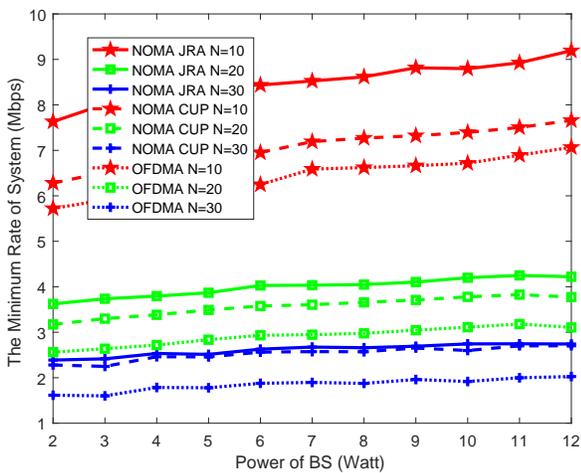}\protect\protect\caption{Minimum user rate for different number of users versus BS power }
\end{figure}

\textcolor{black}{In Fig. 2, we display the sum rates achieved by
various schemes. SR1 JRA and SR2 JRA denote the proposed JRA method
for maximizing the weighted SR and the SR with QoS constraints, while
SR1 DC represents the DC method for maximizing the weighted SR. Meanwhile,
SR1 CUP and SR2 CUP use the CUP method with our optimal power allocation
under the criteria of maximizing the weighted SR and the SR with QoS
constraints, respectively. The number of users is $10$ in this scenario.
As expected, all NOMA schemes (SR1 JRA, SR2 JRA , SR1 CUP, SR2 CUP
and SR1 DC) outperform OFDMA. Moreover, it is also observed that SR1
JRA outperforms SR1 DC. This is because the proposed resource allocation
uses the optimal power allocation while the DC method leads to a suboptimal
power allocation. In addition, both SR1 JRA and SR2 JRA achieve better
performance than SR1 CUP and SR2 CUP, which implies the proposed channel
assignment method is essential to the performance of the NOMA system.}

\begin{figure}[h]
\centering \includegraphics[scale=0.6]{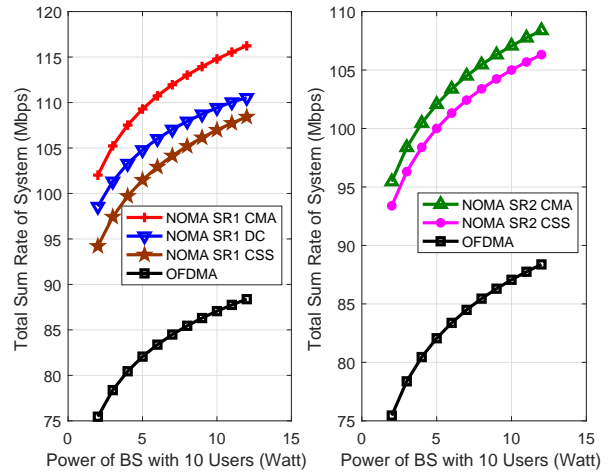}\protect\protect\caption{Sum rate versus BS power }
\end{figure}

\textcolor{black}{Fig. 3 shows the spectral efficiency versus the
number of users. One can observe the similar phenomenon as in Fig.
2. where the spectral efficiency of NOMA outperforms that of OFDMA
and the proposed JRA method leads to higher spectral efficiency than
the DC method and the CUP method.}

\begin{figure}[h]
\centering \includegraphics[scale=0.6]{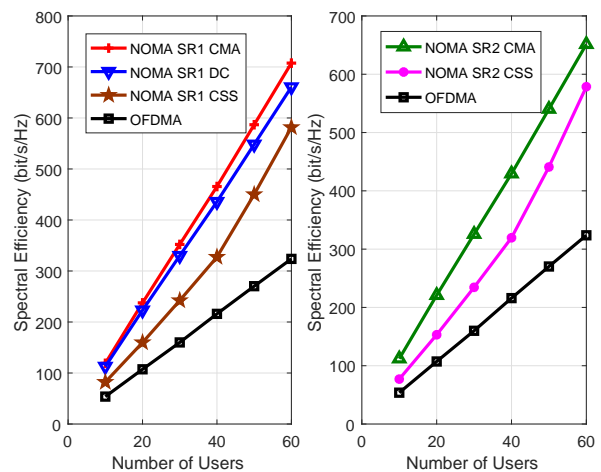}\protect\protect\caption{Spectral efficiency versus the number of users }
\end{figure}

\textcolor{black}{Fig. 4 and Fig. 5 display the EE versus the power
budget of BS and the number of users, respectively. The proposed methods
for maximizing EE with weights or QoS constraints are denoted by EE1
JRA and EE2 JRA, respectively, while EE1 DC denotes the EE maximization
using DC programming. EE1 CUP and EE2 CUP represent the method of
CUP with our optimal power solutions to maximize EE with weights or
QoS constraints, respectively. From Figs. 4 and 5, one can see that
the EE of NOMA is significantly higher than that of OFDMA. Meanwhile,
EE1 JRA achieves better performance than EE1 DC as a result of using
the optimal power allocation. EE1 JRA and EE2 JRA are respectively
better than EE1 CUP and EE2 CUP because the channel assignment is
optimized by the proposed joint optimization method.}

\begin{figure}[h]
\centering \includegraphics[scale=0.6]{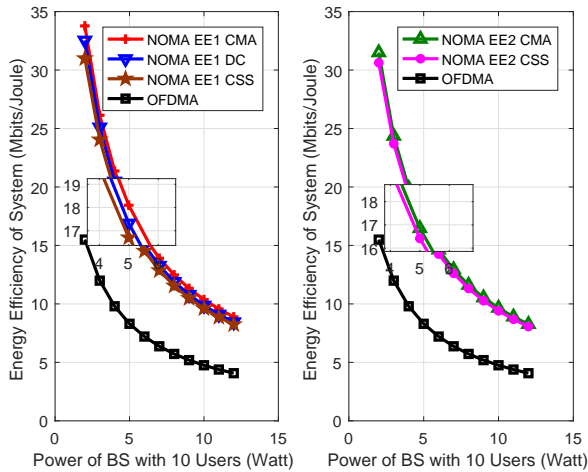}\protect\protect\caption{Energy efficiency versus BS power }
\end{figure}

\begin{figure}
\centering \includegraphics[scale=0.6]{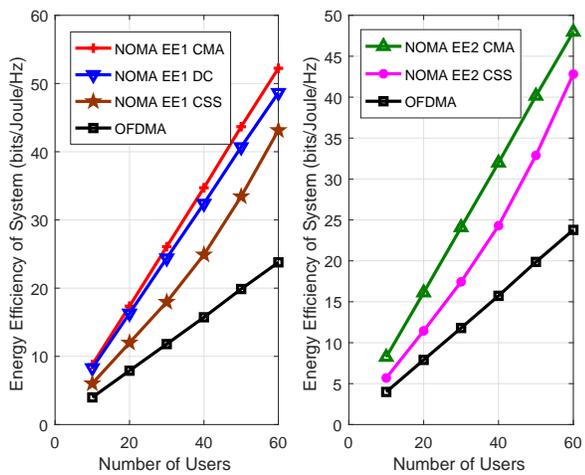}\protect\protect

\caption{Energy efficiency versus the number of users}
\end{figure}

Finally, in Fig. 6, all of our JRA methods are compared to the exhaustive
search (ES). Due to the high complexity of ES, we set the number of
users $N=6$ and the power budget of the BS ranges from $2\textrm{W}$
to $12\textrm{W}$. From Fig. 6, the performance achieved the proposed
methods is very close to the globally optimal value and the maximum
gap is less than $5\%$. Therefore, the proposed joint channel assignment
and power allocation method is able to achieve near-optimal performance
with low complexity.

\begin{figure}
\centering \includegraphics[scale=0.6]{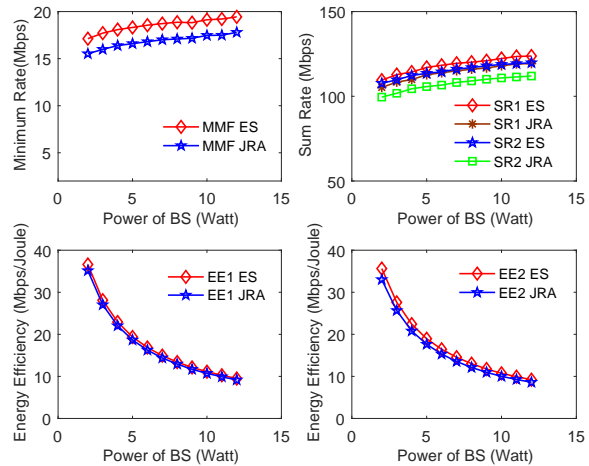}\protect\protect\caption{Comparison with the exhaustive search (ES)}\label{6}
\end{figure}

\section{\label{sec:CONCLUTION}Conclusion}

In this paper, we have studied the power allocation in downlink NOMA
systems to maximize the MMF, sum rate, and EE with weights or QoS
constraints. The optimal power allocation has been characterized in
closed or semi-closed forms for all considered performance criteria.
We have explicitly considered the power order constraints in power
allocation problems and introduced the concept of the SIC-stability
to avoid an equal power allocation on each channel in NOMA systems.
We have proposed an efficient method to jointly optimize the channel
assignment and power allocation in NOMA systems by exploiting the
matching algorithm along with the optimal power allocation. The simulation
results have shown that the proposed joint resource optimization method
achieve near-optimal performance.

\bibliographystyle{IEEEtran}
\bibliography{IEEEabrv,Mybib}

\begin{thebibliography}{10}
\providecommand{\url}[1]{#1}
\csname url@samestyle\endcsname
\providecommand{\newblock}{\relax}
\providecommand{\bibinfo}[2]{#2}
\providecommand{\BIBentrySTDinterwordspacing}{\spaceskip=0pt\relax}
\providecommand{\BIBentryALTinterwordstretchfactor}{4}
\providecommand{\BIBentryALTinterwordspacing}{\spaceskip=\fontdimen2\font plus
\BIBentryALTinterwordstretchfactor\fontdimen3\font minus
  \fontdimen4\font\relax}
\providecommand{\BIBforeignlanguage}[2]{{%
\expandafter\ifx\csname l@#1\endcsname\relax
\typeout{** WARNING: IEEEtran.bst: No hyphenation pattern has been}%
\typeout{** loaded for the language `#1'. Using the pattern for}%
\typeout{** the default language instead.}%
\else
\language=\csname l@#1\endcsname
\fi
#2}}
\providecommand{\BIBdecl}{\relax}
\BIBdecl

\bibitem{andrews2014will}
J.~G. Andrews, S.~Buzzi, W.~Choi, S.~V. Hanly, A.~Lozano, A.~C.~K. Soong, and
  J.~C. Zhang, ``What will 5{G} be?'' \emph{IEEE J. Sel. Areas Commun.},
  vol.~32, no.~6, pp. 1065--1082, Jun. 2014.

\bibitem{jungnickel2014role}
V.~Jungnickel, K.~Manolakis, W.~Zirwas, B.~Panzner, V.~Braun, M.~Lossow,
  M.~Sternad, R.~Apelfrojd, and T.~Svensson, ``The role of small cells,
  coordinated multipoint, and massive {MIMO} in 5{G},'' \emph{IEEE Commun.
  Mag.}, vol.~52, no.~5, pp. 44--51, May. 2014.

\bibitem{7506342}
S.~He, Y.~Huang, J.~Wang, L.~Yang, and W.~Hong, ``Joint antenna selection and
  energy-efficient beamforming design,'' \emph{IEEE Signal Process. Lett.},
  vol.~23, no.~9, pp. 1165--1169, Sept. 2016.

\bibitem{andrews2012femtocells}
J.~G. Andrews, H.~Claussen, M.~Dohler, S.~Rangan, and M.~C. Reed, ``Femtocells:
  {P}ast, present, and future,'' \emph{IEEE J. Sel. Areas Commun.}, vol.~30,
  no.~3, pp. 497--508, Apr. 2012.

\bibitem{7835181}
J.~Wang, W.~Guan, Y.~Huang, R.~Schober, and X.~You, ``{D}istributed
  optimization of hierarchical small cell networks: A {GNEP} framework,''
  \emph{IEEE J. Sel. Areas Commun.}, vol.~35, no.~2, pp. 249--264, Feb. 2017.

\bibitem{7155564}
H.~Wang, J.~Wang, and Z.~Ding, ``Distributed power control in a two-tier
  heterogeneous network,'' \emph{IEEE Trans. Wireless Commun.}, vol.~14,
  no.~12, pp. 6509--6523, Dec. 2015.

\bibitem{6510562}
J.~Wang, D.~Zhu, C.~Zhao, J.~C.~F. Li, and M.~Lei, ``Resource sharing of
  underlaying device-to-device and uplink cellular communications,'' \emph{IEEE
  Commun. Lett.}, vol.~17, no.~6, pp. 1148--1151, Jun. 2013.

\bibitem{07572117}
J.~Wang, Q.~Tang, C.~Yang, R.~Schober, and J.~Li, ``Security enhancement via
  device-to-device communication in cellular networks,'' \emph{IEEE Signal
  Process. Lett.}, vol.~23, no.~11, pp. 1622--1626, Nov. 2016.

\bibitem{wei2016survey}
\BIBentryALTinterwordspacing
Z.~Wei, J.~Yuan, D.~W.~K. Ng, M.~Elkashlan, and Z.~Ding, ``A survey of downlink
  non-orthogonal multiple access for 5{G} wireless communication networks,''
  2016. [Online]. Available: \url{http://arxiv.org/abs/1609.01856.}
\BIBentrySTDinterwordspacing

\bibitem{wang2006comparison}
W.~Peng, X.~Jun, and L.~Ping, ``Comparison of orthogonal and non-orthogonal
  approaches to future wireless cellular systems,'' \emph{IEEE Veh. Technol.
  Mag.}, vol.~1, no.~3, pp. 4--11, Sept. 2006.

\bibitem{saito2013non}
Y.~Saito, Y.~Kishiyama, A.~Benjebbour, T.~Nakamura, A.~Li, and K.~Higuchi,
  ``Non-orthogonal multiple access ({NOMA}) for cellular future radio access,''
  in \emph{Proc. IEEE Veh. Technol. Conf.}, Dresden, Germany, Jun. 2013, pp.
  1--5.

\bibitem{dai2015non}
L.~Dai, B.~Wang, Y.~Yuan, S.~Han, C.~l.~I, and Z.~Wang, ``Non-orthogonal
  multiple access for 5{G}: {S}olutions, challenges, opportunities, and future
  research trends,'' \emph{IEEE Commun. Mag.}, vol.~53, no.~9, pp. 74--81,
  Sept. 2015.

\bibitem{ding2015general}
Z.~Ding, R.~Schober, and H.~V. Poor, ``A general {MIMO} framework for {NOMA}
  downlink and uplink transmission based on signal alignment,'' \emph{IEEE
  Trans. Wireless Commun.}, vol.~15, no.~6, pp. 4438--4454, Jun. 2016.

\bibitem{ding2016application}
Z.~Ding, F.~Adachi, and H.~V. Poor, ``The application of {MIMO} to
  non-orthogonal multiple access,'' \emph{IEEE Trans. Wireless Commun.},
  vol.~15, no.~1, pp. 537--552, Jan. 2016.

\bibitem{2016userandpower}
X.~Zhang, Q.~Gao, C.~Gong, and Z.~Xu, ``User grouping and power allocation for
  {NOMA} visible light communication multi-cell networks,'' \emph{IEEE Commun.
  Lett.}, vol.~21, no.~99, pp. 777 -- 780, Mar. 2016.

\bibitem{7862785}
Z.~Ding, P.~Fan, and H.~V. Poor, ``Random beamforming in millimeter-wave {NOMA}
  networks,'' \emph{IEEE Access}, vol.~5, pp. 7667--7681, Feb. 2017.

\bibitem{zhang2016radio}
S.~Zhang, B.~Di, L.~Song, and Y.~Li, ``Radio resource allocation for
  non-orthogonal multiple access ({NOMA}) relay network using matching game,''
  in \emph{Proc. Int. Conf. Commun.}, K{UL}, {M}alaysia, May. 2016, pp. 1--6.

\bibitem{lei2015joint}
L.~Lei, D.~Yuan, C.~K. Ho, and S.~Sun, ``Joint optimization of power and
  channel allocation with non-orthogonal multiple access for 5{G} cellular
  systems,'' in \emph{Proc. IEEE Global Commun. Conf.}, San Diego, CA, Dec.
  2015, pp. 1--6.

\bibitem{sun2016optimal}
Y.~Sun, D.~W.~K. Ng, Z.~Ding, and R.~Schober, ``Optimal joint power and
  subcarrier allocation for full-duplex multicarrier non-orthogonal multiple
  access systems,'' \emph{IEEE Trans. Commun.}, vol.~65, no.~99, pp. 1077 --
  1091, Mar. 2017.

\bibitem{parida2014power}
P.~Parida and S.~S. Das, ``Power allocation in {OFDM} based {NOMA} systems: {A}
  {DC} programming approach,'' in \emph{Proc. IEEE Globecom Workshops}, Austin,
  {TX}, {USA}, Dec. 2014, pp. 1026--1031.

\bibitem{fang2016energy}
F.~Fang, H.~Zhang, J.~Cheng, and V.~C.~M. Leung, ``Energy-efficient resource
  allocation for downlink non-orthogonal multiple access network,'' \emph{IEEE
  Trans. Commun.}, vol.~64, no.~9, pp. 3722--3732, May. 2016.

\bibitem{hojeij2015resource}
M.~R. Hojeij, J.~Farah, C.~A. Nour, and C.~Douillard, ``Resource allocation in
  downlink non-orthogonal multiple access ({NOMA}) for future radio access,''
  in \emph{Proc. IEEE Veh. Technol. Conf.}, Dresden, Germany, May. 2015, pp.
  1--6.

\bibitem{wang2016power}
C.~L. Wang, J.~Y. Chen, and Y.~J. Chen, ``Power allocation for a downlink
  non-orthogonal multiple access system,'' \emph{IEEE Wireless Commun. Lett.},
  vol.~5, no.~5, pp. 532--535, Oct. 2016.

\bibitem{cui2016novel}
J.~Cui, Z.~Ding, and P.~Fan, ``A novel power allocation scheme under outage
  constraints in {NOMA} systems,'' \emph{IEEE Signal Process. Lett.}, vol.~23,
  no.~9, pp. 1226--1230, Sept. 2016.

\bibitem{choi2016power}
J.~Choi, ``Power allocation for max-sum rate and max-min rate proportional
  fairness in {NOMA},'' \emph{IEEE Commun. Lett.}, vol.~20, no.~10, pp.
  2055--2058, Oct. 2016.

\bibitem{timotheou2015fairness}
S.~Timotheou and I.~Krikidis, ``Fairness for non-orthogonal multiple access in
  5{G} systems,'' \emph{IEEE Signal Process. Lett.}, vol.~22, no.~10, pp.
  1647--1651, Oct. 2015.

\bibitem{zhang2016energy}
Y.~Zhang, H.~M. Wang, T.~X. Zheng, and Q.~Yang, ``Energy-efficient transmission
  design in non-orthogonal multiple access,'' \emph{IEEE Trans. Veh. Technol.},
  vol.~66, no.~99, pp. 2852--2857, Sept. 2016.

\bibitem{ding2015cooperative}
Z.~Ding, M.~Peng, and H.~Poor, ``Cooperative non-orthogonal multiple access in
  5{G} systems,'' \emph{IEEE Commun. Lett.}, vol.~19, no.~8, pp. 1462--1465,
  Aug. 2015.

\bibitem{ding2016impact}
Z.~Ding, P.~Fan, and H.~V. Poor, ``Impact of user pairing on 5{G} nonorthogonal
  multiple-access downlink transmissions,'' \emph{IEEE Trans. Veh. Technol.},
  vol.~65, no.~8, pp. 6010--6023, Aug. 2016.

\bibitem{ali2016dynamic}
M.~S. Ali, H.~Tabassum, and E.~Hossain, ``Dynamic user clustering and power
  allocation for uplink and downlink non-orthogonal multiple access ({NOMA})
  systems,'' \emph{IEEE Access}, vol.~4, pp. 6325--6343, Aug. 2016.

\bibitem{6868214}
Z.~Ding, Z.~Yang, P.~Fan, and H.~V. Poor, ``On the performance of
  non-orthogonal multiple access in 5{G} systems with randomly deployed
  users,'' \emph{IEEE Signal Process. Lett.}, vol.~21, no.~12, pp. 1501--1505,
  Dec. 2014.

\bibitem{dinkelbach1967nonlinear}
W.~Dinkelbach, ``On nonlinear fractional programming,'' \emph{Manage. Sci.},
  vol.~13, no.~7, pp. 492--498, 1967.

\bibitem{2014dynamicmatching}
G.~Hua, J.~Goh, and V.~L.~L. Thing, ``A dynamic matching algorithm for audio
  timestamp identification using the {ENF} criterion,'' \emph{IEEE Trans. Inf.
  Forensics Secur.}, vol.~9, no.~7, pp. 1045--1055, July. 2014.

\bibitem{2017two-sidematching}
L.~Gao, L.~Duan, and J.~Huang, ``Two-sided matching based cooperative spectrum
  sharing,'' \emph{IEEE Trans. Mob. Comput.}, vol.~16, no.~2, pp. 538--551,
  Feb. 2017.

\bibitem{2016Scableand}
Y.~Gao, Z.~Qin, Z.~Feng, Q.~Zhang, O.~Holland, and M.~Dohler, ``Scalable and
  reliable {I}o{T} enabled by dynamic spectrum management for {M}2{M} in
  {LTE}-{A},'' \emph{IEEE Internet Things J.}, vol.~3, no.~6, pp. 1135--1145,
  Dec. 2016.

\bibitem{7511620}
H.~Zhang, D.~K. Zhang, W.~X. Meng, and C.~Li, ``User pairing algorithm with
  {S}{I}{C} in non-orthogonal multiple access system,'' in \emph{Proc. Int.
  Conf. Commun.}, K{UL}, {M}alaysia, May. 2016, pp. 1--6.

\end{thebibliography}

\begin{IEEEbiography}[{{\includegraphics[width=1in,height=1.25in,clip,keepaspectratio]{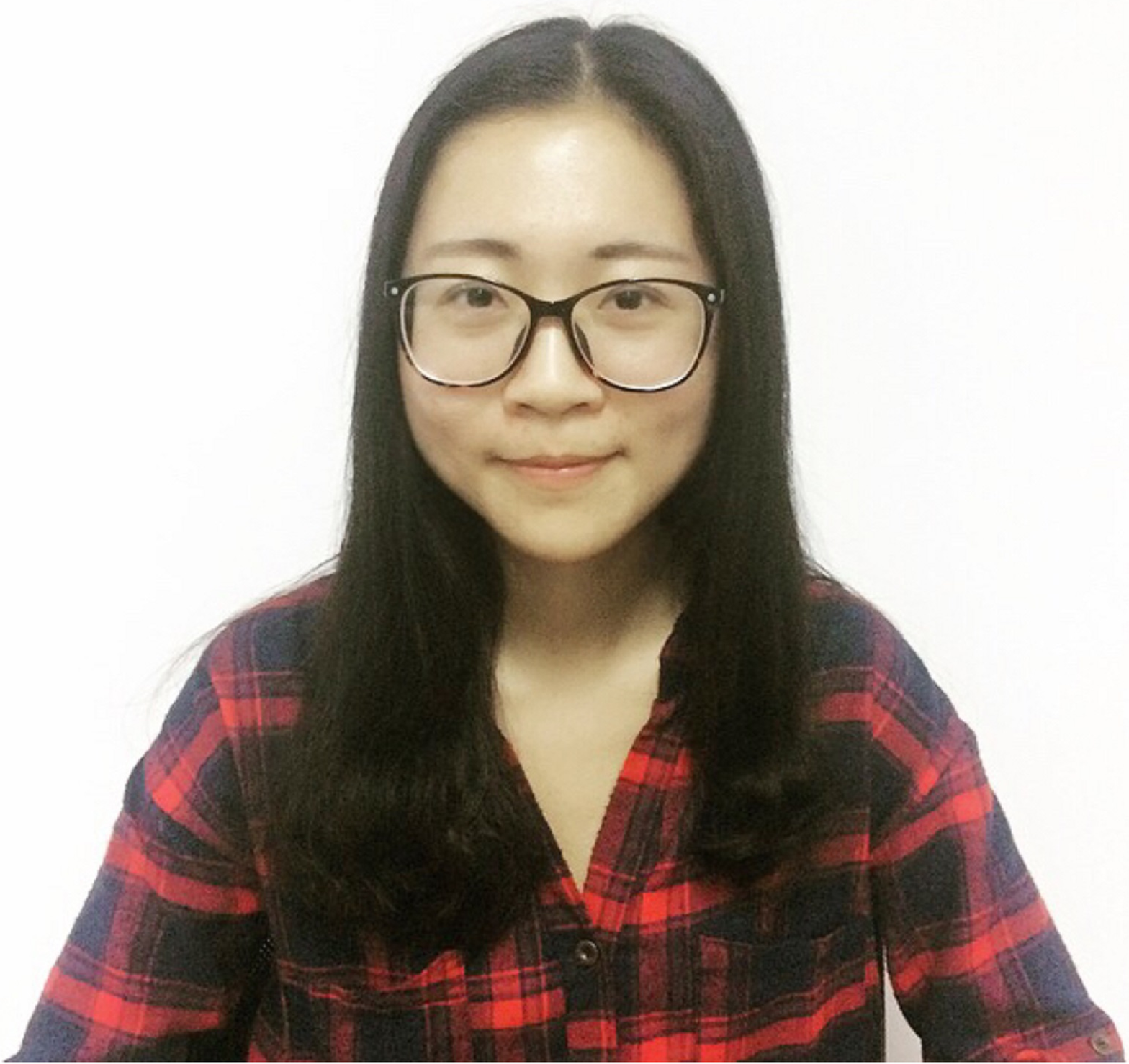}}}]{Jianyue Zhu}(S'17)
 received the B.S. degree from Nanjing Agricultural University, Nanjing, China, in
2015. She is currently working towards the Ph.D. degree in information
and communication engineering at the School of Information Science
and Engineering, Southeast University, Nanjing, China. Her current
research interests include non-orthogonal multiple access and optimization
theory.
\end{IEEEbiography}

\begin{IEEEbiography}[{\includegraphics[width=1in,height=1.25in,clip,keepaspectratio]{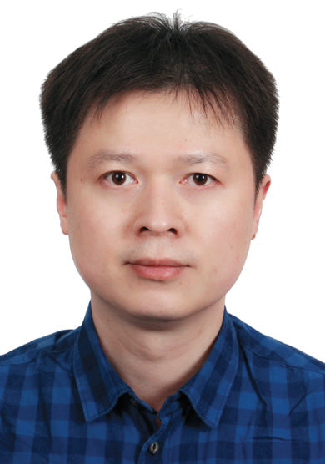}}]{Jiaheng Wang}(M'10,
SM'14) received the Ph.D. degree in electronic and computer engineering from the Hong Kong University of Science and Technology, Kowloon, Hong Kong, in 2010, and the B.E. and M.S. degrees from the Southeast University, Nanjing, China, in 2001 and 2006, respectively.

He is a Full Professor at the National Mobile Communications Research Laboratory (NCRL), Southeast University, Nanjing, China. From 2010 to 2011, he was with the Signal Processing Laboratory, KTH Royal Institute of Technology, Stockholm, Sweden. He also held visiting positions at the Friedrich Alexander University Erlangen-N\"{u}rnberg, N\"{u}rnberg, Germany, and the University of Macau, Macau. His research interests are mainly on optimization in communication systems and wireless networks.

Dr. Wang has published more than 80 articles on international journals and conferences. He serves as an Associate Editor for the IEEE Signal Processing Letters. He is a recipient of the Humboldt Fellowship for Experienced Researchers, and a recipient of the Best Paper Award in WCSP 2014.
\end{IEEEbiography}

\begin{IEEEbiography}[{{\includegraphics[width=1in,height=1.25in,clip,keepaspectratio]{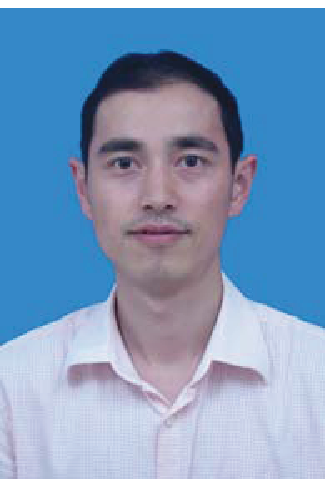}}}]{Yongming Huang}(M'10)
received the B.S. and M.S. degrees from Nanjing University,
Nanjing, China, in 2000 and 2003, respectively, and the Ph.D. degree
in electrical engineering from Southeast University, Nanjing, in 2007.

Since March 2007, he has been a faculty member with the School of
Information Science and Engineering, Southeast University, China,
where he is currently a full professor. In 2008-2009, he visited the
Signal Processing Laboratory, School of Electrical Engineering, Royal
Institute of Technology (KTH), Stockholm, Sweden. His current research
interests include MIMO wireless communications, cooperative wireless
communications and millimeter wave wireless communications. He has
published over 200 peer-reviewed papers, hold over 40 invention patents,
and submitted over 10 technical contributions to IEEE standards. Since
2012, he has served as an Associate Editor for the IEEE Transactions
on Signal Processing, EURASIP Journal on Advances in Signal Processing,
and EURASIP Journal on Wireless Communications and Networking.
\end{IEEEbiography}

\begin{IEEEbiography}[{{\includegraphics[width=1in,height=1.25in,clip,keepaspectratio]{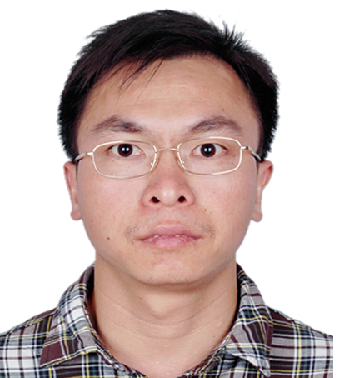}}}]{Shiwen he}(M'14)
received the M.S. degree from Chengdu University of Technology, Chengdu, China, and the Ph.D. degree in information
and communication engineering from Southeast University, Nanjing, China, in 2009 and 2013, respectively.

From 2013 to 2015, he was a Postdoctoral Researcher with the State Key Laboratory of Millimeter
Waves, Department of Radio Engineering, Southeast University. Since October 2015, he has been with the School of Information Science
and Engineering, Southeast University. He has authored or coauthored
over 80 technical publications, hold over 16 granted invention patents,
and submitted over 20 technical proposals to IEEE standards. From
December 2015, he serves as a technical sub-editor for IEEE 802.11aj.
His research interests include multiuser MIMO communication, cooperative
communications, energy efficient communications, millimeter wave communication,
and optimization theory.
\end{IEEEbiography}

\begin{IEEEbiography}[{\includegraphics[width=1in,height=1.25in,clip,keepaspectratio]{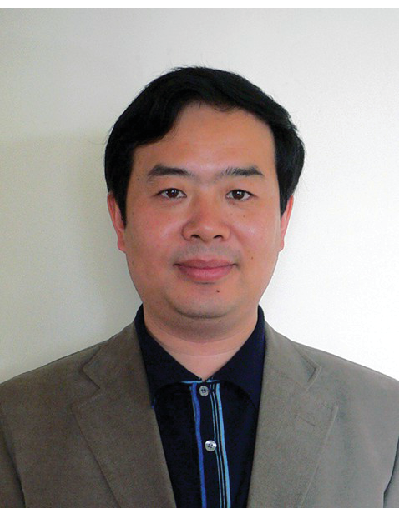}}]{Xiaohu You}(F'11)
was born in 1962. He received his Master and Ph.D. degrees from Southeast
University, Nanjing, China, in Electrical Engineering in 1985 and
1988, respectively. Since 1990, he has been working with National
Mobile Communications Research Laboratory at Southeast University,
where he held the rank of director and professor. His research interests
include mobile communication systems, signal processing and its applications.
He has contributed over 100 IEEE journal papers and 2 books in the
areas of adaptive signal processing, neural networks and their applications
to communication systems. From 1999 to 2002, he was the Principal
Expert of the C3G Project, responsible for organizing China\textquoteright s
3G Mobile Communications R\&D Activities. From 2001-2006, he was the
Principal Expert of the China National 863 Beyond 3G FuTURE Project.
Since 2013, he has been the Principal Investigator of China National
863 5G Project.

Prof. You was the general chairs of IEEE WCNC 2013 and IEEE VTC
2016. Now he is Secretary General of the FuTURE Forum, vice Chair
of China IMT-2020 Promotion Group, vice Chair of China National Mega
Project on New Generation Mobile Network. He was the recipient of
the National First Class Invention Prize in 2011, and he was selected
as IEEE Fellow in same year due to his contributions to development
of mobile communications in China.
\end{IEEEbiography}

\begin{IEEEbiography}[{{\includegraphics[clip,width=1in,height=1.25in,keepaspectratio]{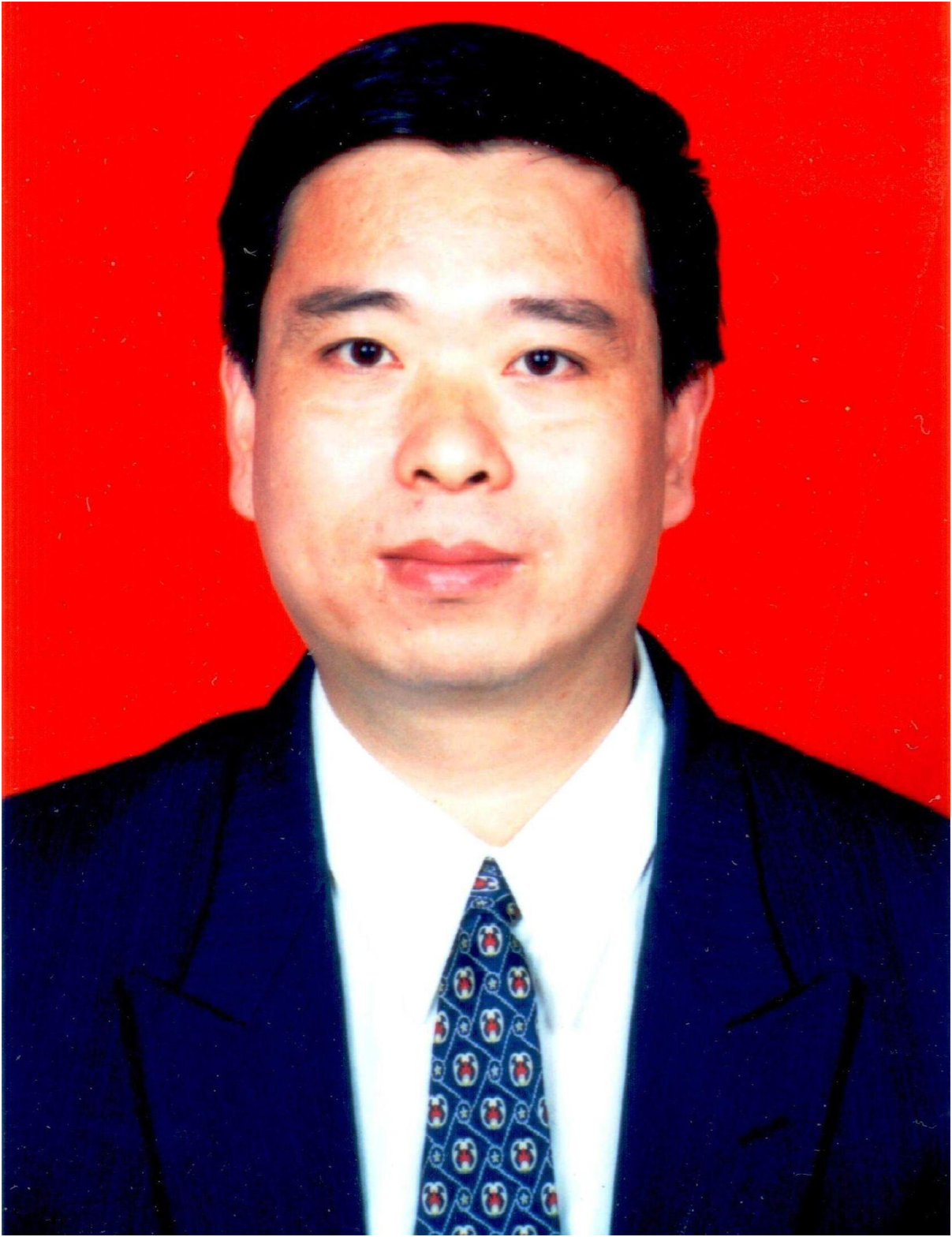}}}]{Luxi Yang}(M'96)
received the M.S. and Ph.D. degree in electrical engineering
from the Southeast University, Nanjing, China, in 1990 and 1993, respectively.

Since 1993, he has been with the Department of Radio Engineering,
Southeast University, where he is currently a full professor of information
systems and communications, and the Director of Digital Signal Processing
Division. His current research interests include signal processing
for wireless communications, MIMO communications, cooperative relaying
systems, and statistical signal processing. He has authored or co-authored
of two published books and more than 100 journal papers, and holds
30 patents. Prof. Yang received the first and second class prizes
of science and technology progress awards of the state education ministry
of China in 1998, 2002 and 2014. He is currently a member of Signal
Processing Committee of Chinese Institute of Electronics.
\end{IEEEbiography}

\end{document}